\documentclass[letterpaper]{article}
\usepackage{aaai19}
\usepackage{times}  
\usepackage{helvet}  
\usepackage{courier}  
\usepackage{url}  
\usepackage{graphicx}  
\usepackage{amsmath,amsthm,amssymb}
\usepackage[ruled, vlined]{algorithm2e}
\usepackage{amsfonts,amssymb,amstext,amsmath}
\usepackage{tikz}
\usetikzlibrary{arrows, petri, topaths, automata}
\usepackage{caption}  
\usepackage{tkz-berge}
\usepackage[textsize=tiny]{todonotes}
\usepackage{xspace}
\usepackage{thmtools,thm-restate}
\usepackage{graphicx}
\usepackage{enumerate}
\usepackage{floatrow}
\usepackage{fancyvrb}
\usepackage{booktabs}
\newcommand{\bigo}[1]{\mathbf{\mathcal{O}} \left ( #1 \right)}

\newcommand{\expect}[1]{\mathbb{E}  \left [ #1 \right ]}

\newcommand{\opt}{\textsc{opt}\xspace}

\newcommand{\absl}[1]{\left \lvert #1 \right \rvert}

\newcommand{\mymut}{\mathsf{pmut}_{\beta}\xspace}

\newcommand{\mut}[1]{\mymut}

\newcommand{\app}{$\varepsilon$-approximation\xspace}

\newcommand{\greedy}{\textsc{Greedy}\xspace}

\newcommand{\probl}{Problem \ref{problem}\xspace}

\newtheorem{definition}{Definition}
\newtheorem{reduction}{Reduction}
\newtheorem{theorem}{Theorem}
\newtheorem{lemma}{Lemma}

\newtheorem{proposition}{Proposition}

\newtheorem{problem}{Problem}
\newcommand{\ignore}[1]{}

\begin{document}
%
\title{Greedy Maximization of Functions with Bounded Curvature under Partition Matroid Constraints}

\author{Tobias Friedrich\textsuperscript{1}, Andreas G\"obel\textsuperscript{1}, Frank Neumann\textsuperscript{2}, Francesco Quinzan\textsuperscript{1} and Ralf Rothenberger\textsuperscript{1}\\
\textsuperscript{1}{Chair of Algorithm Engineering, Hasso Plattner Institute, Potsdam, Germany}\\
\textsuperscript{2}{Optimisation and Logistics, School of Computer Science, The University of Adelaide, Adelaide, Australia}\\}

\maketitle
\begin{abstract}
We investigate the performance of a deterministic \greedy algorithm for the problem of maximizing functions under a partition matroid constraint. We consider non-monotone submodular functions and monotone subadditive functions. Even though constrained maximization problems of monotone submodular functions have been extensively studied, little is known about greedy maximization of non-monotone submodular functions or monotone subadditive functions.

We give approximation guarantees for \greedy on these problems, in terms of the \emph{curvature}. We find that this simple heuristic yields a strong approximation guarantee on a broad class of functions.

We discuss the applicability of our results to three real-world problems: Maximizing the determinant function of a positive semidefinite matrix, and related problems such as the maximum entropy sampling problem, the constrained maximum cut problem on directed graphs, and combinatorial auction games.

We conclude that \greedy is well-suited to approach these problems. Overall, we present evidence to support the idea that, when dealing with constrained maximization problems with bounded curvature, one needs not search for (approximate) monotonicity to get good approximate solutions.
\end{abstract}

\section{Introduction} 
Submodular functions capture the notion of diminishing returns, i.e. the more we acquire the less our marginal gain will be. This notion occurs frequently in the real world, thus, the problem of maximizing a submodular function finds applicability in a plethora of scenarios. Examples of such scenarios include: maximum cut problems~\cite{GoemansW95}, combinatorial auctions~\cite{MaeharaKSTK17}, facility location~\cite{citeulike:2277727}, problems in machine learning~\cite{EDFK17}, coverage functions~\cite{DBLP:journals/jmlr/KrauseSG08}, online shopping~\cite{TSK17}. As such, the literature on submodular functions contains a vast number of results spanning over three decades.

Formally, a set function $f\colon 2^V\rightarrow \mathbb{R}$ is \emph{submodular} if for all $U,W\subseteq V$, $f(U) + f(W) \geq f(U \cup W) + f(U \cap W)$. As these functions come from a variety of applications, in this work we will assume that, given a set $U\subseteq V$, the value $f(U)$ is returned from an oracle. This is a reasonable assumption as in most applications $f(U)$ can be computed efficiently. Often in these applications, a realistic solution is subject to some constraints. Among the most common constraints are Matroid and Knapsack constraints ---see \cite{DBLP:conf/stoc/LeeMNS09}. From these families of constraints the most natural and common type of constraints are uniform matroid constraints also known as cardinality constraints. Optimizing a submodular function given $k$ as a cardinality constraint is equivalent to finding a set $U$, with $|U|\leq k$, that maximizes $f(U)$. In this paper we consider submodular maximization under partition matroid constraints. These constraints are in the intersection of matroid and knapsack constaints and generalize uniform matroid constraints. In partition matroid constraints we are given a collection $B_1, \dots, B_k $ of disjoint subsets of~$V$ and integers $d_1\dots d_k$. Every feasible solution to our problem must then include at most $d_i$ elements from each set $B_i$. Submodular maximization under partition matroid constraints is considered in various applications, e.g. see~\cite{Lin:2010:MSV:1857999.1858133,5995589}.

The classical result of~\cite{citeulike:2277727} shows that a greedy algorithm achieves a $1/2$ approximation ratio when maximizing monotone submodular functions under partition matroid constraints. \cite{DBLP:journals/mor/NemhauserW78} showed that no-polynomial time algorithm can achieve a better approximation ratio than $(1-1/e)$. Many years later \cite{DBLP:journals/siamcomp/CalinescuCPV11} where able to achieve this upper bound using a randomized algorithm. Recently \cite{DBLP:journals/corr/abs-1807-05532} achieved a deterministic $0.5008$-approximation ratio by derandomizing search heuristics.

The previous approximation ratios can be further improved when assuming that the rate of change of the marginal values of $f$ is bounded. This is expressed by the curvature $\alpha$ of a function as in Definition~\ref{def:generalized_curvature}. The results of \cite{DBLP:journals/dam/ConfortiC84,V10} show that a continuous greedy algorithm gives a $\frac{1}{\alpha} ( 1 - e^{ -\alpha})$ approximation when maximizing a monotone submodular function under a matroid constraint.  Finally, \cite{DBLP:conf/icml/BianB0T17} show that the deterministic greedy algorithm achieves a $\frac{1}{\alpha} ( 1 - e^{ -\alpha})$ approximation when maximizing monotone submodular functions of curvature $\alpha$, but only under cardinality constraints.

All of the aforementioned approximation results rely on the fact that $f$ is monotone, i.e. $f(S) \leq f(T)$ for all $S\subseteq T$. In practice submodular functions such as maximum cut, combinatorial auctions, sensor placement, and experimental design need not be monotone. To solve such problems using simple greedy algorithms, often assumptions are made that the function $f$ is monotone or that $f$ is under some sense ``close'' to being monotone. Practical problems that are solved using greedy algorithms under such assumptions can be found in many articles such as \cite{DBLP:conf/icml/BianB0T17,DBLP:conf/icml/DasK11,DBLP:conf/nips/LawrenceSH02,DBLP:journals/jair/SinghKGK09}.

In this article we show that the greedy algorithm finds a $\frac{1}{\alpha} ( 1 - e^{ -\alpha})$-approximation in $\bigo{dn}$ oracle evaluations, for the problem of maximizing a submodular function subject to uniform matroid constraints (Theorem~\ref{thm:det_greedy_algorithm}). Furthermore, we derive similar approximation guarantees for the partition matroid case.

Additionally, we extend the results on monotone submodular functions to another direction, to the class of monotone subadditive functions. Subadditivity is a natural property assumed to hold for functions evaluating items sold in combinatorial auctions \cite{DBLP:conf/soda/BhawalkarR11,DBLP:conf/sigecom/Assadi17}. Formally, we say that a set function $f\colon 2^V\rightarrow \mathbb{R}$ is \emph{subadditive} if for all $U,W\subseteq V$, $f(U) + f(W) \geq f(U \cup W)$. We show (Theorem~\ref{thm:submodular_function_positive2}) that the greedy algorithm achieves a $\frac{1}{\alpha} ( 1 - e^{ \alpha^2-\alpha})$\footnote{In the case of a monotone function it always holds $\alpha \in [0, 1]$.} approximation ratio when optimizing monotone subadditive functions with curvature~$\alpha$ under uniform matroid constraints. As in the case of submodular functions, we extend these results to the case of a partition matroid constraint.

We motivate our results by considering three real world applications. The first application we consider is to maximize the logarithm of determinant functions. In this setting we are given a matrix $\mathcal{P}$ and we want to find the submatrix~$\mathcal{A}$ of $\mathcal{P}$ with the largest determinant, where~$\mathcal{A}$ satisfies matroid partition constraints. This problem appears in a variety of real world settings. In this article, as a real world example of this application, we compute the sensor (thermometer) placement across the world that maximizes entropy, subject to a cardinality constraint and subject to a partition matroid constraint where the partitions of the data sets are countries.

Our second application is the problem of finding the maximum directed cut of a graph, under partition matroid constraints. The cut function of a graph is known to be submodular and non-monotone in general \cite{DBLP:journals/siamcomp/FeigeMV11}. We show how to bound the curvature of the cut function with respect to the maximum degree. We also run experiments on this setting, showing that in most graphs of our dataset the deterministic greedy algorithm finds the actual optimal solution. Thus \greedy seems to perform well on non-monotone submodular functions in practice.

Finally, the third application is computing the social welfare of a subadditive combinatorial auction. We show that the social welfare is also a subadditive function and its curvature is bounded by the maximum curvature of the utility functions.

\section{Preliminary Definitions and Algorithms}
\label{subsection:problem_description}
\subsection{Problem description.}
We study the following optimization problem.
\begin{problem}
\label{problem}
Let $f\colon 2^V \longrightarrow \mathbb{R}_{\geq 0}$ be a non-negative function\footnote{We always assume that $f$ is normalized, that is $f(\emptyset) = 0$.} over a set $V$ of size $n$, let $B_1, \dots, B_k $ be a collection of disjoint subsets of $V$, and let $d_i$ integers s.t. $1 \leq d_i \leq \absl{B_i}, \ \forall i \in [k]$. We consider the maximization problem
\begin{equation*}
\max_{S\subseteq V} \left \{ f(S) : \absl{S \cap B_i}\leq d_i, \ \forall i \in [k]   \right \}.
\end{equation*}
\end{problem}
Note that the problem of maximizing $f$ under a cardinality constraint is a special case of the above, where $k = 1$ and $B_1 = V$.

We evaluate the quality of an approximation of a global maximum as follows. Let $U \subseteq V$ be a feasible solution to \probl. We say that $U$ is an \app if $f(U)/f(\opt) \geq \varepsilon $, where \opt is the optimal solution set. We often refer to the value $f(U)$ as the $f$-value of $U$. 

In this paper, we evaluate run time in the black-box oracle model: We assume that there exists an oracle that returns the corresponding $f$-value of a solution candidate, and we estimate the run time, by counting the total number of calls to the evaluation oracle.

To simplify the exposition, throughout our analyses, we always assume that the following reduction holds.
\begin{reduction}
\label{reduction}
For \probl we may assume $cd_i \leq \absl{B_i}$ for all $i = 1, \dots, k$, for an arbitrary constant $c> 0$. Moreover, we may assume that there exists a set $D_i \subseteq B_i$ of size $d_i$ s.t. $f(S) = f(S\setminus D_i)$ for all $S \subseteq V$, for all $i = 1, \dots, k$. 
\end{reduction}

\subsection{Algorithms.}

\begin{algorithm*}[t]
    \textbf{input:} a function $f\colon 2^V\longrightarrow \mathbb{R}_{\geq 0}$\;
    $\qquad \ \ \ $ disjoint subsets $B_1, \dots, B_k \subseteq V$\;
    $\qquad \ \ \ $ integers $d_1, \dots, d_k$ s.t. $0 \leq d_i \leq \absl{B_i}, \ \forall i \in [k]$\;
    \textbf{output:} an approximate global maximum $S$ of $f$ s.t. $\absl{S\cap B_i} \leq d_i, \ \forall i \in [k]$;
	\caption{The \greedy algorithm.}
 	$S \gets \emptyset$\;
 	\While{$\absl{S} \leq \sum_{i = 1}^k d_i$}{
	let $\omega \in V$ maximizing $f(S\cup \{ \omega \}) - f(S)$ and s.t. $\absl{(S\cup \{ \omega \})\cap B_i} \leq d_i, \ \forall i \in [k]$\;
            $S \gets S \cup \{\omega \}$\;
   	}
    \textbf{return} $S$\;
    \label{alg:greedy}
\end{algorithm*}
\greedy is the simple discrete greedy algorithm that appears in Algorithm~\ref{alg:greedy}. Starting with the empty set, \greedy iteratively adds points that maximize the marginal values with respect to the already found solution. This algorithm is a mild generalization of the simple deterministic greedy algorithm due to Nemhauser and Wolsey \cite{DBLP:journals/mor/NemhauserW78}. 

\subsection{Notation.}
\label{sec:notation}
For any  non-negative function $f\colon 2^V \longrightarrow \mathbb{R}_{\geq 0}$ and any two subsets $S, \Omega \subseteq V$, we define the \emph{marginal value} of $S$ with respect to $\Omega$ as $\rho_{\Omega} (S) = f(S \cup \Omega) - f(S)$.

We denote with $B_1, \dots, B_k$ disjoint subsets of $V$ and with $d_1, \dots, d_k$ their respective sizes, as in the problem description section. We denote with $d$ the sum $\sum_{j = 1}^k d_j$, and we define $\overline{d} = \inf_i d_i$. We denote with $D$ the subset of ``dummy" elements as in Reduction~\ref{reduction}, and we denote with $\opt$ any solution to \probl, such that $\opt \cap D = \emptyset$.

We let $S_t$ be a solution found by \greedy at time step $t$ and we denote with $\rho_t$ the marginal value $\rho_t = f(S_t) - f(S_{t - 1})$. We use the convention $\rho_0 = f(\emptyset)$. We define $\omega_t = S_{t} \setminus S_{t - 1}$.

\section{Curvature}

In this paper we give approximation guarantees in terms of the \emph{curvature}. 
Intuitively, the curvature is a parameter that bounds the maximum rate with which a function changes. As our functions~$f$ map sets to positive reals, i.e. $f\colon 2^V\longrightarrow \mathbb{R}_{\geq 0}$, we say that $f$ has curvature $\alpha$ if the value $f(S\cup \{\omega\}) - f(S)$ does not change by a factor larger than $1 - \alpha$ when varying $S$. This parameter was first introduced by \cite{DBLP:journals/dam/ConfortiC84} to beat the $(1 - e^{-1})$-approximation barrier of monotone submodular functions. Formally we use the following definition of curvature, relaxing the definition of \emph{greedy} curvature \cite{DBLP:conf/icml/BianB0T17}. 
\begin{definition}[Curvature]
\label{def:generalized_curvature}
Consider a non-negative function $f\colon 2^V \longrightarrow \mathbb{R}$ as in \probl. The curvature is the smallest scalar $\alpha$ s.t.
\[
\rho_\omega ((S \cup \Omega )\setminus \{\omega\}) \geq (1 - \alpha)\rho_\omega (S \setminus \{\omega\}),
\]
for all $S, \Omega \subseteq V$ and $\omega \in S\setminus \Omega$.
\end{definition}
Note that $\alpha \geq 0$. We say that a function $f$ has \emph{positive curvature} if $\alpha \leq 1$. Otherwise, we say that $f$ has \emph{negative curvature}. Note that a function is monotone iff. it has positive curvature. We remark that the curvature is invariant under multiplication by a positive scalar. In other words, if a function $f$ has curvature $\alpha$, then any function $c f$ has curvature $\alpha$, for all $c>0$. Moreover, the following simple result holds.
\begin{proposition}
\label{prop:sum_curvature}
Let $f, g \colon 2^V \longrightarrow \mathbb{R}_{\geq 0}$ be non-negative functions with curvature $\alpha_1, \alpha_2$ respectively. Then the curvature $\alpha$ of the function $f + g$ is upper-bounded as $\alpha \leq \sup_i \alpha_i$.
\end{proposition}
In the case of a submodular function, it is possible to give a simple characterization of Definition \ref{def:generalized_curvature}. In fact, one can easily prove the following.
\begin{proposition}
\label{prop:curvature}
Let $f\colon 2^V \longrightarrow \mathbb{R}_{\geq 0}$ be a submodular function with curvature $\alpha$, as in \probl. Then,
\[
\alpha \leq 1 - \min_{\{S\subseteq V,\ \omega \in S\}}\frac{f(S) - f(S\setminus \{ \omega \})}{f(\omega) - f(\emptyset)},
\]
for all subsets $S \subseteq V$.
\end{proposition}

\section{Approximation Guarantees}
\label{sec:approx_greedy}
We give approximation guarantees for \greedy on \probl, when optimizing a (non-monotone) submodular function with bounded curvature $\alpha$. Our proof technique generalizes the results of \cite{DBLP:journals/dam/ConfortiC84} to non-monotone functions $f$ by utilizing the notion of curvature. We have the following theorem.
\begin{theorem}
\label{thm:det_greedy_algorithm}
Let $f$ be a submodular function with curvature $\alpha$. \greedy is a $\frac{1}{\alpha} ( 1 - e^{-\alpha \overline{d}/d})$-approximation algorithm for \probl with run-time in $\bigo{dn}$.
\end{theorem}
Note that if $f$ is monotone, then our approximation guarantee matches the approximation guarantee of Conforti and Cornu{\'{e}}jols, which is known to be nearly optimal \cite{DBLP:journals/dam/ConfortiC84,V10}, in the uniform matroid case. Furthermore, in the non-monotone case our lower-bound may yield significant improvement over state-of-the-art known bounds \cite{DBLP:conf/soda/BuchbinderFNS14,BF18}. Particularly, we beat the $1/e$-approximation barrier on functions with curvature $\alpha \leq 2.49375$ and the $1/2$-approximation barrier on functions with curvature $\alpha \leq 1.59362$. 

We give some approximation guarantee for \greedy, assuming that the function $f$ is monotone subadditive. Our proof method further generalizes the proof of \cite{DBLP:journals/dam/ConfortiC84}.  The following theorem holds.
\begin{theorem}
\label{thm:submodular_function_positive2}
Let $f$ be a monotone subadditive function with curvature $\alpha \in [0, 1]$, and suppose that $f(\emptyset) = 0$. Then \greedy is a $\frac{1}{\alpha} ( 1 - e^{(\alpha^2 -\alpha )\overline{d}/d})$-approximation algorithm for \probl with run-time in $\bigo{dn}$.
\end{theorem}
To our knowledge, this is the first approximation guarantee for the simple \greedy maximizing a monotone subadditive function under partition matroid constraints.

\section{Applications}
\label{sec:applications}
\subsection{Maximizing the logarithm of determinant functions.}
\label{sec:positive_semidefinite}

An $n\times n$ matrix $\mathcal{P}$ is positive definite if $\mathcal{P}$ is symmetric and all its eigenvalues $\lambda_1, \dots, \lambda_n$ are strictly greater than $0$.
Consider such an $n \times n$ positive definite matrix $\mathcal{P}$. The \emph{determinant function} $\mbox{det}_\mathcal{P} : \{0, 1\}^n \rightarrow \mathbb{R}_{\geq 0}$, with input an array $\mathbf{x} \in \{0, 1\}^n$, returns the determinant of the square sub-matrix of $\mathcal{P}$ indexed by $\mathbf{x}$. We search for a sub-matrix of $\mathcal{P}$ that satisfies a partition or a cardinality constraint, and such that $\log \mbox{det}$ is maximal. 

Variations of this setting can be found in informative vector machines \cite{DBLP:conf/nips/LawrenceSH02} and in maximum entropy sampling problems \cite{DBLP:journals/jmlr/KrauseSG08}.

The constrained problem of maximizing $\log \mbox{det}_\mathcal{P}$ is studied in the context of maximizing submodular functions under a single matroid constraint with a continuous greedy and non-oblivious local search in \cite{DBLP:journals/mor/SviridenkoVW17}.

The problem of maximizing $\mbox{det}_\mathcal{P}$ under a cardinality constraint is studied in \cite{DBLP:conf/icml/BianB0T17}, when $\mathcal{P}$ is a matrix of the form $\mathcal{P} = \mathcal{I} + \sigma \Sigma$, with $\mathcal{I}$ the $n\times n$ identity matrix, $\Sigma$ a positive semidefinite matrix, and $\sigma > 0$ a scalar. In this case, the function $\det_\mathcal{P}$ is monotone, supermodular, and the submodularity ratio can be estimated in terms of the eigenvalues. Note that a matrix of the form $\mathcal{I} + \sigma \Sigma$ always has eigenvalues $\lambda_j \geq 1$.\\
We study the problem of maximizing $\mbox{det}_\mathcal{P}$ under a partition matroid constraint, assuming that $\mathcal{P}$ is positive definite with eigenvalues $\lambda_j \geq 1$.  We show that in this case the simple greedy algorithm is sufficient to obtain a nearly-optimal approximation guarantee. If $\log \mbox{det}_{\mathcal{P}}$ is non-constant, using Proposition~\ref{prop:curvature} we can upper bound the activity by $\alpha \leq 1 - 1/\lambda$, where $\lambda$ is the largest eigenvalue of $\mathcal{P}$ \cite{DBLP:journals/mor/SviridenkoVW17}. Thus, \greedy gives a $(1 - e^{1/\lambda -1})/\left (1 - 1/\lambda \right )$-approximation for \probl when $f = \log \mbox{det}_{\mathcal{P}}$ with runtime in $\bigo{nd}$. We do not assume that the eigenvalues are such that $\lambda_j > 1$, so our analysis applies to monotone as well as non-monotone functions. For instance, consider the function $\log \mbox{det}_{\mathcal{A}}$ with
\[
\mathcal{A} = \left ( 
\begin{array}{cc}
\delta & \sqrt{\delta - 1}\\
\sqrt{\delta - 1} & 1
\end{array}
\right )
\]
for all $\delta > 1$. In this case, the function $\log \det_{\mathcal{A}}$ is neither monotone, nor approximately monotone \cite{DBLP:conf/stoc/LeeMNS09,DBLP:journals/jmlr/KrauseSG08}. \greedy, nevertheless, finds a $\left (1 - 1/\delta \right )(1 - e^{1/\delta -1})$-approximation of the global optimum under uniform matroid constraints.

We can further generalize this result to more complex functions, by means of Proposition \ref{prop:sum_curvature}. For instance, let $f$ be the entropy function of a Gaussian process, as defined in \eqref{max:entropy}. Then the function $f$ is the sum of a linear term $((1 + \ln (2\pi))/2) \absl{S}$ and $1/2 \ln \det_{\Sigma}(S)$, for a positive semidefinite matrix $\Sigma$ with eigenvalues $\lambda_j \geq 1$. This function is submodular, because both terms are submodular. Moreover, the linear term has curvature $\alpha = 0$, and the function $1/2 \ln \det_{\Sigma}(S)$ has curvature $1 - 1/\lambda$, with $\lambda$ the largest eigenvalue of $\Sigma$. Hence, we combine Theorem \ref{thm:det_greedy_algorithm} with Proposition \ref{prop:sum_curvature} to conclude that \greedy is a $(1 - 1/\lambda)(1 - e^{1/\lambda - 1})$-approximation algorithm for \probl in the uniform case, with $f$ the entropy as in \eqref{max:entropy}. Note that our analysis does not require monotonicity, and it holds for matrices such as $\Sigma = \mathcal{A}$. 

\subsection{Finding the maximum directed cut of a graph.}
\begin{figure}[t]
\caption{We consider a bipartite graph $G = (V, E)$ of order $n$ and size $n$, with partitions labeled as $A$ and $B$. In this example, there's only one node in $A$, and $n- 1$ nodes in $B$. Moreover, there's only one edge from $A$ to $B$, whereas there is one edge from each node in $B$ to $A$. Since all nodes in $A$ and $B$ have equal $f$-value, then \greedy may output $A$ as a solution to the maximum cut under uniform constrain of size $d$. This yields a $1/d$-approximation of the global optimum.
}
\label{fig:bad_example}
\includegraphics[width=\linewidth]{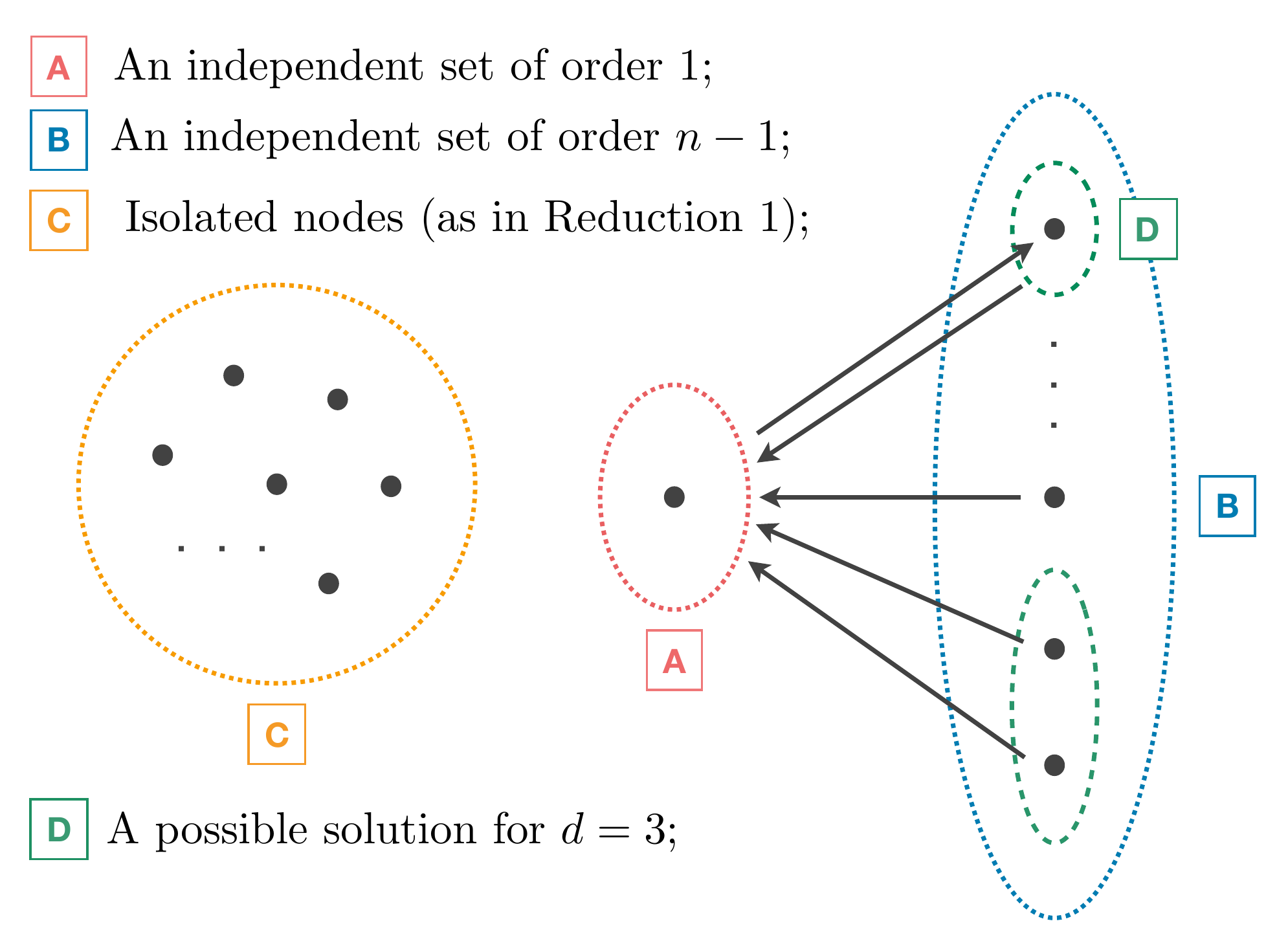}
\end{figure}

Let $G = (V, E)$ be a graph with $n$ vertices and $m$ edges, together with a non-negative weight function $w\colon E \rightarrow \mathbb{R}_{\geq 0}$. We consider the problem of finding a subset $U \subseteq V$ of nodes such that the sum of the weights on the outgoing edges of $U$ is maximal. 
This problem is the maximum directed cut problem known to be $\mathsf{NP}$-complete. We consider a constrained version of this problem, as in \probl. We consider both directed and undirected graphs. We first define the cut function as follows.

Let $G = (V, E)$ and $w$ be as above.  For each subset of nodes $U \subseteq V$, consider the set $C(U) = \{ (e_1, e_2) \in E \colon e_1 \in U \mbox{ and } e_2 \notin U  \}$ of the edges leaving $U$. We define the \emph{cut function} $f\colon 2^V \rightarrow \mathbb{R}_{\geq 0}$ with  
$f(U) = \sum_{e \in C(U)}w(e).$

The constrained maximum directed cut problem can be approached by maximizing the cut function under a uniform cardinality constraint. Since we require the weights to be non-negative, this function is also non-negative. As noted in \cite{DBLP:journals/siamcomp/FeigeMV11}, the cut function is always submodular and, in general, non-monotone.

Denote with $\Delta^+$ the maximum out-degree of $G$, i.e. the maximum degree when counting \emph{outgoing} edges and denote with $\Delta^-$ the maximum in-degree of $G$, obtained by counting the \emph{incoming} edges only. Then from Proposition \ref{prop:curvature} the curvature of the corresponding cut function is upper-bounded as $\alpha \leq 1 + \frac{\Delta^- }{\Delta^+}$. When $G$ is undirected, $\Delta^- = \Delta^+$ and, therefore, $\alpha \leq 2$. Thus, Theorem \ref{thm:det_greedy_algorithm} yields that \greedy is a $1/2 (1 - e^{-2})$-approximation algorithm for the constrained maximum cut problem. This approximation guarantee improves as $d$ decreases.

When $G$ is a directed graph the approximation guarantee can drop to $1/d$. Consider a bipartite graph $G=(V, E)$ with $n$ vertices and $n$ edges of weight~1 described as follows (see Figure~\ref{fig:bad_example}). Let $A, B$ be the partitions of $V$. $A$ contains exactly one node and $B$ contains $n - 1$ nodes. The unique vertex of $A$ has exactly one outgoing edge to a vertex in $B$. Each vertex in $B$ has an outgoing edge to the only vertex of $A$. When maximizing the cut function of this graph $G$ under the special case of cardinality constraint $d$, the optimal solution consists of $d$ nodes in $B$. \greedy though, may output $A$ as a possible solution, which yields only a $1/d$-approximation of the optimal solution. In this case the curvature is $\alpha \geq d$. However, we show experimentally that the \greedy performs well on a variety of real-world networks. We remark that in real-world networks the degree $\Delta_+$ is expected to grow in the problem size \cite{DBLP:journals/siamrev/Newman03,article_barbasi}.
\subsection{Social welfare in combinatorial auctions.}
We consider combinatorial auctions with $n$ players competing for $m$ items, where the items can have different values for each player. Moreover, the value of each item for a player may depend on the particular combination of items allocated to that player. For any given player $i = 1, \dots, n$, the value of a combination of items is expressed by the \emph{utility function} $u_i \colon 2^{[m]} \rightarrow \mathbb{R}_{\geq 0}$. The objective of the social welfare problem (SW) is to find disjoint sets $S_1, \dots, S_n$ maximizing the total welfare $\sum_{i = 1}^n u_i(S_i)$. Following  \cite{DBLP:conf/soda/BhawalkarR11}, we make the following natural assumptions on all utility functions:
\begin{enumerate}
\item $u_i(\emptyset) = 0$;
\item $u_i(U\cup T) \leq u_i(U) + u_i(T) \mbox{ for all } U, T \subseteq M$;
\item $u_i(U) \leq u_i(T) \mbox{ for all } U \subseteq T \subseteq M$.
\end{enumerate}
Since an explicit description of a utility function may require exponential space, we assume the existence of an \emph{oracle} that returns the values of $u_i$ for sets of items. In the literature, various oracle models have been considered \cite{DBLP:journals/mor/DobzinskiNS10}. We study the case where for each utility function $u_i$ and any set of items $S\subseteq M$ there exists an oracle that returns the value $u_i(S)$. We refer to this setting as \emph{value oracle model}. We remark that in the context of combinatorial auctions, the utility function $u_i$ of a player is unknown to other players. Thus players may choose not to reveal the true value of the cost functions. In this setting, however, we assume all players to be \emph{truthful}.

\begin{figure*}[ht]
\includegraphics[width=0.5\linewidth]{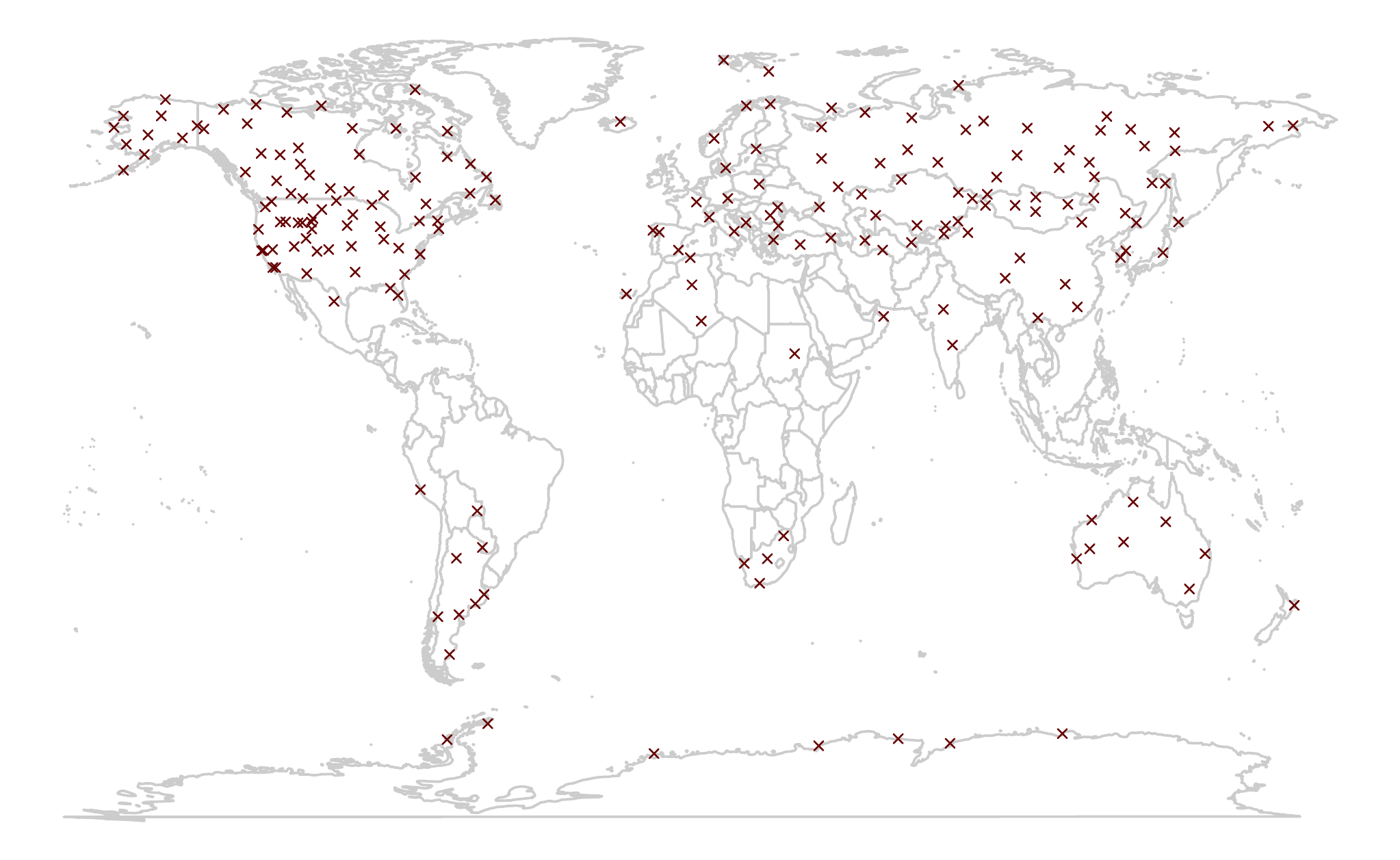}\hfill
\includegraphics[width=0.5\linewidth]{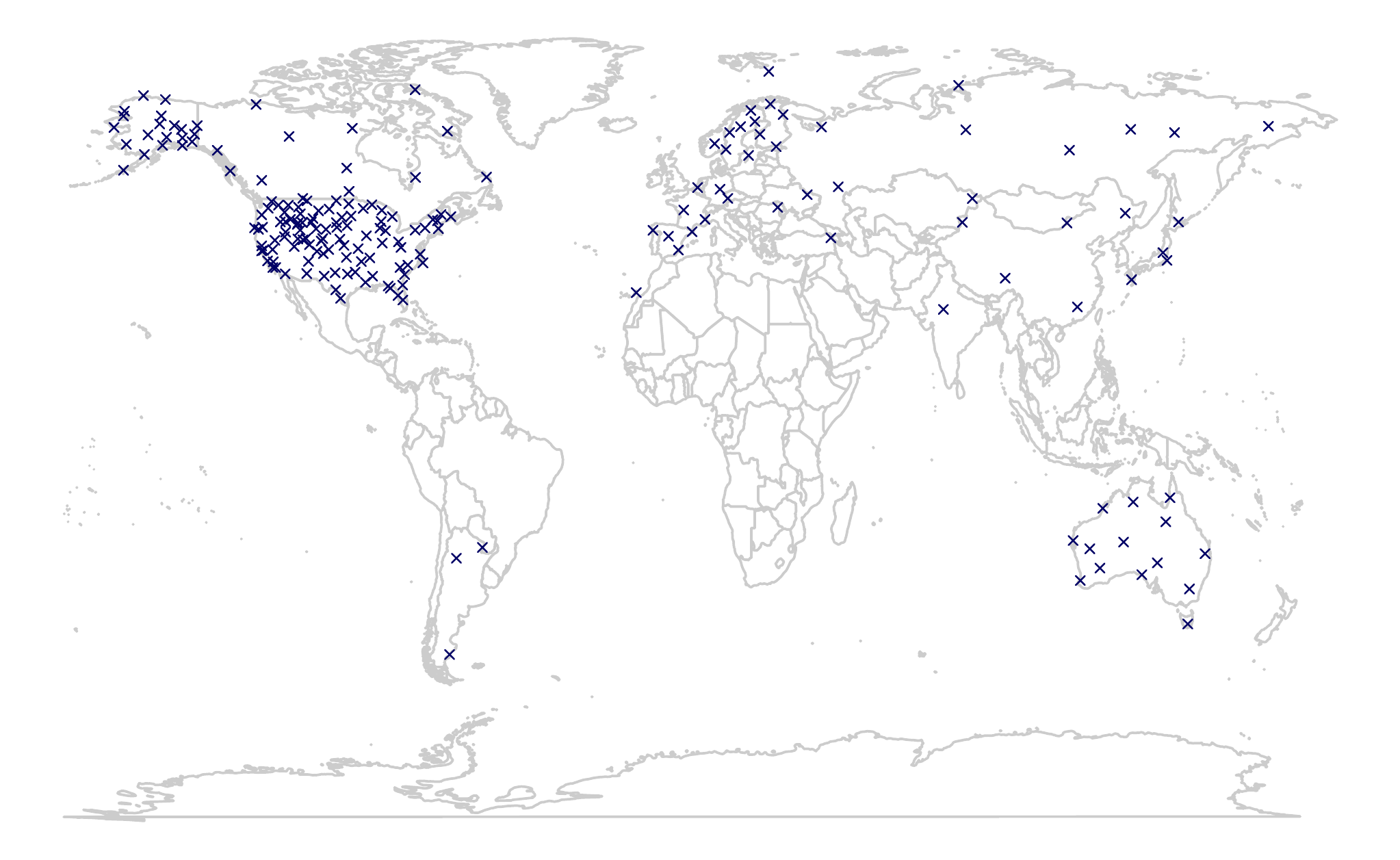}
\caption{A visualization of the solution found by \greedy for $d = 10\%$ in the case of a uniform constraint (left), and a partition constraint by countries (right). In both case, a solution is obtained by maximizing the entropy as given in \eqref{max:entropy}. The covariance matrix $\Sigma$ for all possible locations is displayed in Figure \ref{fig:covariance}. We observe that in the case of a cardinality constraint, the informative stations tend to be spread out, whereas in the partition constraint by countries they tend to be grouped in a few areas. We remark that in the original dataset stations are not distributed uniformly among countries. 
\label{fig:visualize_results}
}
\end{figure*}

We formalize SW as a maximization problem under a partition matroid constraint, following \cite{DBLP:journals/toc/FeigeV10}. For a given set of items $M$ and $n$ players, we define a ground set $X = [n] \times M$. The elements of $X$ are copies of the items in $M$. For each player we require a copy of each item in $M$. For each player $i$ we define a mapping $\pi_i :2^X \longmapsto 2^n$ that assigns copies of items to respective players. In other words, for each set $I \times S \subseteq X$ it holds $$\pi_i(I\times S) = \{\omega \in M \colon (i, \omega) \in I\times S \}.$$ Given  utility functions $u_1, \dots, u_n$, the social welfare problem (SW) consists then of maximizing the following function
\[
f(S) = \sum_{i = 1}^n u_i (\pi_i(S)).
\]
We note that the function $f$ is subadditive, monotone and such that $f(\emptyset) = 0$. In this setting a feasible solution $S$ cannot assign the same item to multiple players. Thus, if we define $B_m = [n] \times \{m\}$, for all items $m \in M$, then a feasible solution $S$ must fulfill the constrain $\absl{S \cup B_m} \leq 1$ for all $m \in M$. Thus, maximizing $f$ in the above setting is equivalent to maximizing a monotone function under a partition matroid constraint. 

Consider SW with $n$ players, $m$ items, and utility functions $u_1, \dots u_n$. Denote with $\alpha_i$ the curvature of each utility function $u_i$. Then the function $f$ has curvature $\alpha  \leq \max_{i} \alpha_i$, by iteratively applying Proposition \ref{prop:sum_curvature}. We can now apply Theorem~\ref{thm:submodular_function_positive2} and conclude that \greedy is a $1/\alpha \left (1 - e^{(\alpha^2 - \alpha)/M} \right )$-approximation algorithm for SW in the value oracle model.
\begin{figure}[t]
\caption{(a) A visualization of the monthly temperature variations of three time series, with particularly high variance. Each series corresponds to a unique station ID. We model each variation series as a Gaussian distribution.\\
(b) Optimal solution found by \greedy for a uniform constraint and a partition matroid constraint by countries. The $f$-value of each set of stations is the entropy \eqref{max:entropy}, with $\Sigma$ the covariance matrix of variation series as in (a) (see Figure \ref{fig:covariance}). 
}
\label{fig:results}
\includegraphics[width=\linewidth]{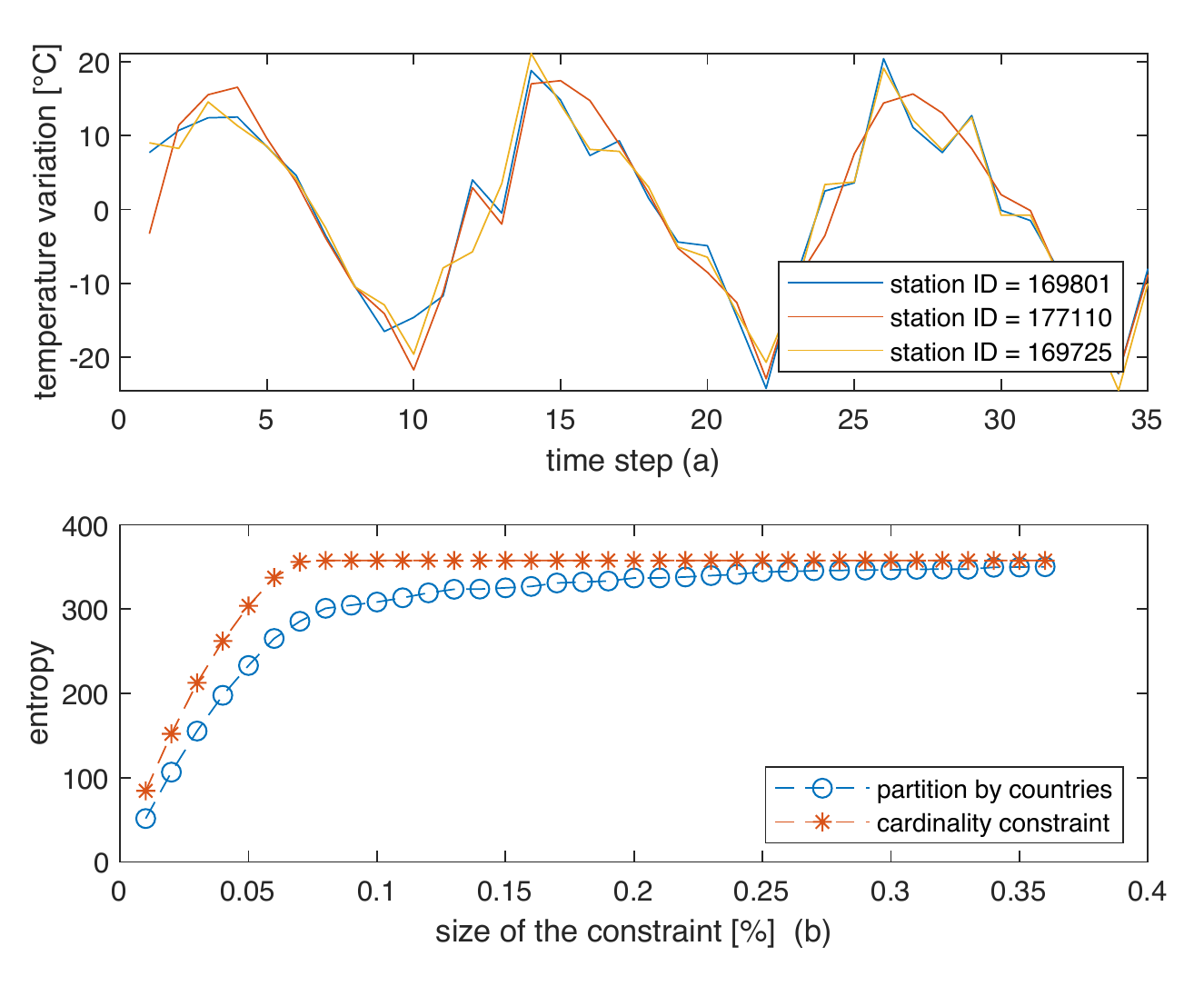}
\end{figure}
\begin{figure}[t]
\caption{A visualization of the covariance matrix $\Sigma$ of time series available in the Berkley Earth climate dataset. We consider stations that have full available reports between years 2015-2017, for a total of $2736$ stations. We consider the variation between average monthly temperatures of each time series. Each entry of this matrix is computed by taking the sample covariance as in \eqref{sample_covariance}.
}
\label{fig:covariance}
\includegraphics[width=\linewidth]{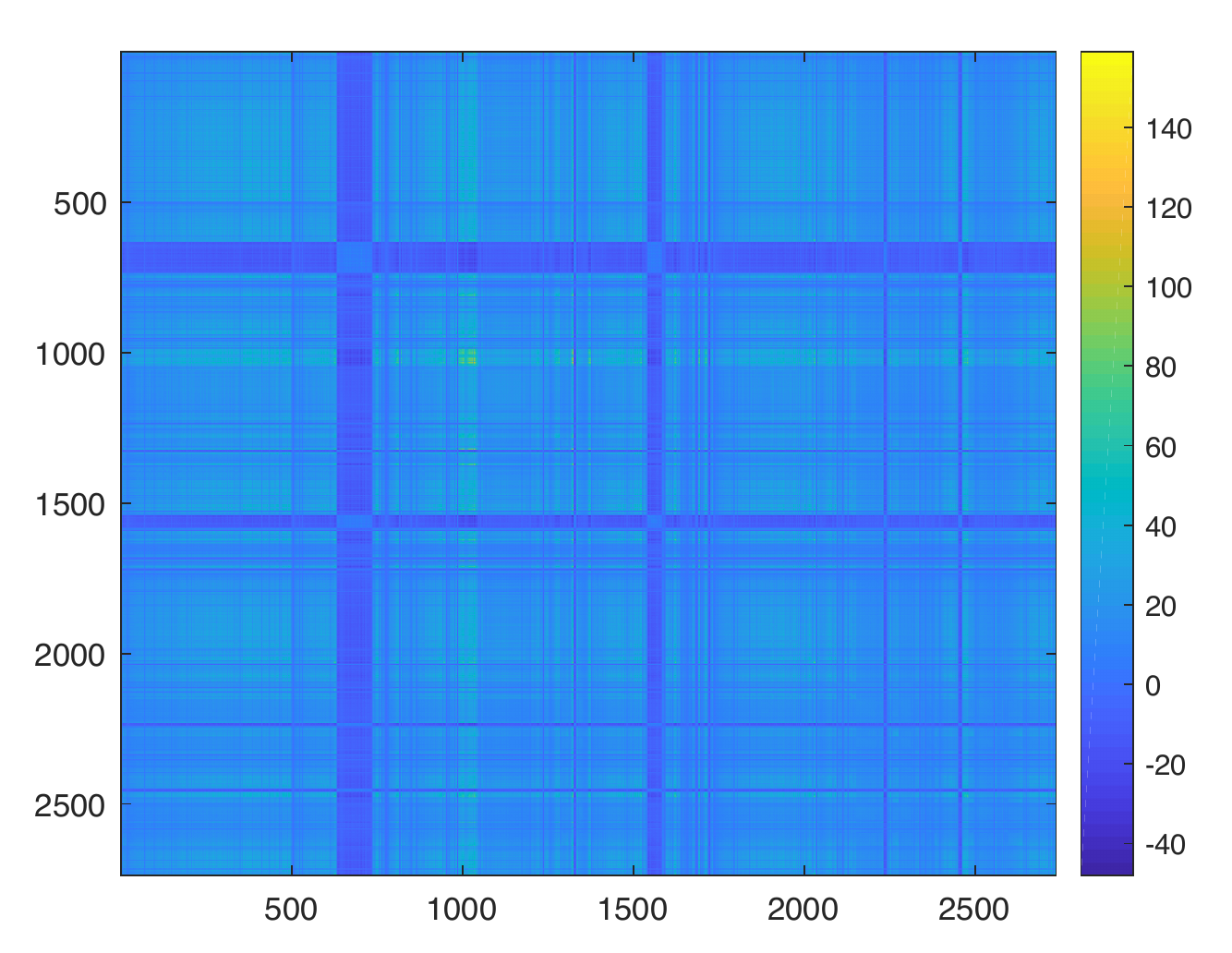}
\end{figure}

\section{Experiments}
\label{sec:experiments}

\subsection{The maximum entropy sampling problem.}

In this set of experiments we study the following problem: Given a set of random variables (RVs), find the most informative subset of variables, subject to a side constraint as in \probl. This setting finds a broad spectrum of applications, from Bayesian experimental design \cite{doi:10.1111/1467-9868.00225}, to monitoring spatio-temporal dynamics \cite{DBLP:journals/jair/SinghKGK09}.

We consider the Berkley Earth climate dataset \footnote{http://berkeleyearth.org/data/}. This dataset combines $1.6$ billion temperature reports from $16$ preexisting data archives, for over $39.000$ unique stations. For each station, we consider a unique time series for the \emph{average} monthly temperature. We always consider time series that span between years 2015-2017. This gives us a total of $2736$ time series, for unique corresponding stations. The code is available at [removed for review].

We study the problem of searching for the most informative sets of time series under various constraints, based on these observations. Given a time series $\mathbf{X} = \{X_t\}_t$ we study the corresponding variation series $\overline{\mathbf{X}} = \{\overline{X}_t\}_t$ defined as $\overline{X}_t = X_t - X_{t - 1}$. A visualization of time series $\overline{\mathbf{X}}$ is given in Figure \ref{fig:results}(a).

We compute the covariance matrix $\Sigma$ between series $\overline{\mathbf{X}}$, $\overline{\mathbf{Y}}$, the entries of which are defined as
\begin{equation}
\label{sample_covariance}
\mathsf{cov}(\overline{\mathbf{X}}, \overline{\mathbf{Y}}) = \frac{1}{m - 1} \sum_{t = 1}^m (\overline{X}_{t} - \expect{\overline{\mathbf{X}}})(\overline{Y}_{t} - \expect{\overline{\mathbf{Y}}}),
\end{equation}
with $m = 35$ the length of each series. A visualization of the covariance matrix $\Sigma$ is given in Figure \ref{fig:covariance}. 

Assuming that the joint probability distribution is Gaussian, we proceed by maximizing the \emph{entropy}, defined as
\begin{equation}
\label{max:entropy}
f(S) = \frac{1 + \ln (2\pi)}{2} \absl{S} + \frac{1}{2} \ln \mbox{det}_\Sigma (S)
\end{equation}
for any indexing set $S \subseteq \{0, 1\}^n$. 

\begin{figure*}[t]
\includegraphics[width=\linewidth]{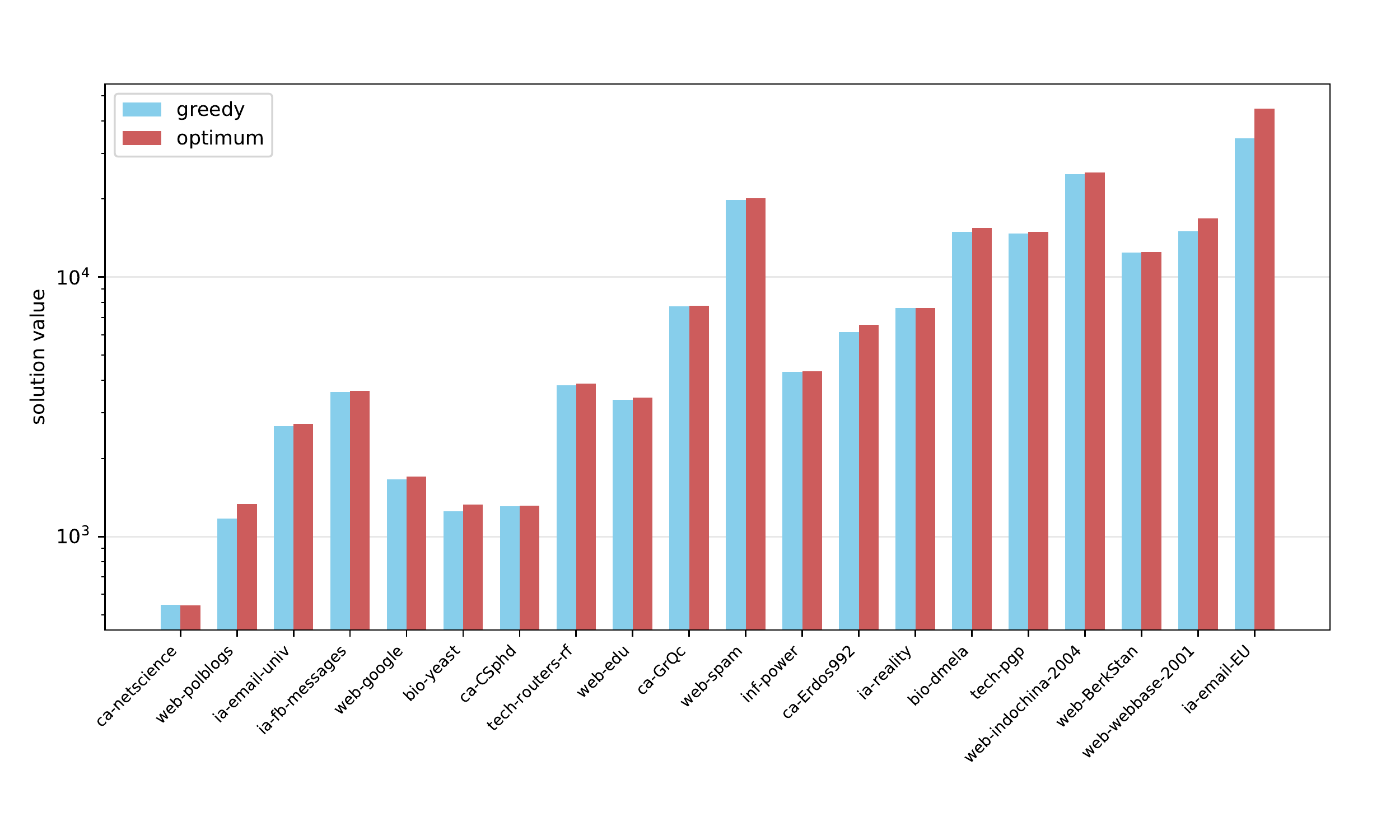}
\caption{Visualization of the optimal solution and the solution found by \greedy for the maximum directed cut problem.
The input graphs are ordered by increasing number of vertices from left to right.}
\label{fig:max-cut}
\end{figure*}

We consider two types of constraints. In a first set of experiments we consider the problem of maximizing the entropy as in \eqref{max:entropy}, under a cardinal constraint only. Specifically, given a parameter $d$, the goal is to find a subset of time series that maximizes the entropy, of size at most $d$ of all available data. We also consider a more complex constraint: Find a subset of time series that maximizes the entropy, and s.t. it contains at most $d$ of all available data of each country. The latter constraint is a partition matroid constraint, where each subset $B_i$ consists of all data series measured by stations in a given country.

A summary of the results is displayed in Figure \ref{fig:results}(b). We observe that in both cases the entropy quickly evolves to a stationary local  optimum,  indicating that a relatively small subset of stations is sufficient to explain the random variations between monthly observations in the model. We observe that the \greedy reaches similar approximation guarantees in both cases.  We remark that the \greedy finds a nearly optimal solution under a cardinality constraint, assuming that the entropy is (approximately) monotone \cite{DBLP:journals/jmlr/KrauseSG08}.

In Figure \ref{fig:visualize_results} we display solutions found by \greedy for the cardinality and partition matroid constraint, with $d = 10\%$. 

We observe that in the case of a cardinality constraint, the sensors spread across the map;  in the case of a partition matroid constraint sensors tend to be placed unevenly. We remark that in the original data set, some countries have a much higher density of stations than others.

\subsection{Finding the maximum directed cut of a graph.}
In this set of experiments we study the performance of \greedy for maximum directed cut in unweighted graphs. We compare these results with the optimal solutions, which we found via an Integer Linear Program solved with the state-of-the-art solver Gurobi~\cite{gurobi}.
The experiments were conducted on $20$ instances from Network Repository~\cite{network-repo}.

Figure~\ref{fig:max-cut} displays the quality of the solution found by \greedy compared to the optimal solution, in the unconstrained case. One can see that in most cases the greedy solution is very close to the optimum. This suggests that \greedy might perform well on real-world instances. We remark that the solution quality is expected to increase as the size of a possible constraint lowers. Thus, \greedy is expected to perform even better in the constrained case. 

Theorem~\ref{thm:det_greedy_algorithm} implies that this might be due to the curvature $\alpha$ of these graphs. However, we find that the solution quality of \greedy is much better than the theoretical upper bound on the curvature. 

\section{Conclusion}

In this paper we consider the problem of maximizing a function with bounded curvature under a single partition matroid constraint.

We derive approximation guarantees for the simple \greedy algorithm (see Algorithm \ref{alg:greedy}) on those problems, in the case of a (non-monotone) submodular function, and a monotone subadditive function (see Theorem \ref{thm:det_greedy_algorithm} and Theorem \ref{thm:submodular_function_positive2}). We observe that the lower bound on the approximation guarantee is asymptotically tight in the case of a submodular function.

We discuss three applications of our setting, and show experimentally that \greedy is suitable to approach these problems.

\section{Acknowledgement}
This research has been supported by the Australian Research Council (ARC) through grant DP160102401 and the German Science Foundation (DFG) through grant FR2988 (TOSU).
\bibliographystyle{aaai}\bibliography{bibliography}

\begin{thebibliography}{}

\bibitem[\protect\citeauthoryear{Albert and Barabasi}{2001}]{article_barbasi}
Albert, R., and Barabasi, A.-L.
\newblock 2001.
\newblock Statistical mechanics of complex networks.
\newblock {\em Reviews of Modern Physics} 74.

\bibitem[\protect\citeauthoryear{Assadi}{2017}]{DBLP:conf/sigecom/Assadi17}
Assadi, S.
\newblock 2017.
\newblock Combinatorial auctions do need modest interaction.
\newblock In {\em {EC}},  145--162.

\bibitem[\protect\citeauthoryear{Bhawalkar and
  Roughgarden}{2011}]{DBLP:conf/soda/BhawalkarR11}
Bhawalkar, K., and Roughgarden, T.
\newblock 2011.
\newblock Welfare guarantees for combinatorial auctions with item bidding.
\newblock In {\em Proc. of {SODA}},  700--709.

\bibitem[\protect\citeauthoryear{Bian \bgroup et al\mbox.\egroup
  }{2017}]{DBLP:conf/icml/BianB0T17}
Bian, A.~A.; Buhmann, J.~M.; Krause, A.; and Tschiatschek, S.
\newblock 2017.
\newblock Guarantees for greedy maximization of non-submodular functions with
  applications.
\newblock In {\em Proc. of {ICML}},  498--507.

\bibitem[\protect\citeauthoryear{Buchbinder and Feldman}{2018}]{BF18}
Buchbinder, N., and Feldman, M.
\newblock 2018.
\newblock Deterministic algorithms for submodular maximization problems.
\newblock {\em {ACM} Transactions on Algorithms} 14(3):32:1--32:20.

\bibitem[\protect\citeauthoryear{Buchbinder \bgroup et al\mbox.\egroup
  }{2014}]{DBLP:conf/soda/BuchbinderFNS14}
Buchbinder, N.; Feldman, M.; Naor, J.; and Schwartz, R.
\newblock 2014.
\newblock Submodular maximization with cardinality constraints.
\newblock In {\em Proc. of {SODA}},  1433--1452.

\bibitem[\protect\citeauthoryear{Buchbinder, Feldman, and
  Garg}{2018}]{DBLP:journals/corr/abs-1807-05532}
Buchbinder, N.; Feldman, M.; and Garg, M.
\newblock 2018.
\newblock Deterministic {(1/2} + {\(\epsilon\)})-approximation for submodular
  maximization over a matroid.
\newblock {\em CoRR} abs/1807.05532.

\bibitem[\protect\citeauthoryear{C{\u{a}}linescu \bgroup et al\mbox.\egroup
  }{2011}]{DBLP:journals/siamcomp/CalinescuCPV11}
C{\u{a}}linescu, G.; Chekuri, C.; P{\'{a}}l, M.; and Vondr{\'{a}}k, J.
\newblock 2011.
\newblock Maximizing a monotone submodular function subject to a matroid
  constraint.
\newblock {\em {SIAM} Journal of Computing} 40(6):1740--1766.

\bibitem[\protect\citeauthoryear{Conforti and
  Cornu{\'{e}}jols}{1984}]{DBLP:journals/dam/ConfortiC84}
Conforti, M., and Cornu{\'{e}}jols, G.
\newblock 1984.
\newblock Submodular set functions, matroids and the greedy algorithm: Tight
  worst-case bounds and some generalizations of the rado-edmonds theorem.
\newblock {\em Discrete Applied Mathematics} 7(3):251--274.

\bibitem[\protect\citeauthoryear{Cornuejols, Fisher, and
  Nemhauser}{1977}]{citeulike:2277727}
Cornuejols, G.; Fisher, M.~L.; and Nemhauser, G.~L.
\newblock 1977.
\newblock {Location of Bank Accounts to Optimize Float: An Analytic Study of
  Exact and Approximate Algorithms}.
\newblock {\em Management Science} 23(8):789--810.

\bibitem[\protect\citeauthoryear{Das and Kempe}{2011}]{DBLP:conf/icml/DasK11}
Das, A., and Kempe, D.
\newblock 2011.
\newblock Submodular meets spectral: Greedy algorithms for subset selection,
  sparse approximation and dictionary selection.
\newblock In {\em Proc. of {ICML}},  1057--1064.

\bibitem[\protect\citeauthoryear{Dobzinski, Nisan, and
  Schapira}{2010}]{DBLP:journals/mor/DobzinskiNS10}
Dobzinski, S.; Nisan, N.; and Schapira, M.
\newblock 2010.
\newblock Approximation algorithms for combinatorial auctions with
  complement-free bidders.
\newblock {\em Mathematics of Operations Research} 35(1):1--13.

\bibitem[\protect\citeauthoryear{Elenberg \bgroup et al\mbox.\egroup
  }{2017}]{EDFK17}
Elenberg, E.~R.; Dimakis, A.~G.; Feldman, M.; and Karbasi, A.
\newblock 2017.
\newblock Streaming weak submodularity: Interpreting neural networks on the
  fly.
\newblock In {\em Advances in Neural Information Processing Systems},
  4047--4057.

\bibitem[\protect\citeauthoryear{Feige and
  Vondr{\'{a}}k}{2010}]{DBLP:journals/toc/FeigeV10}
Feige, U., and Vondr{\'{a}}k, J.
\newblock 2010.
\newblock The submodular welfare problem with demand queries.
\newblock {\em Theory of Computing} 6(1):247--290.

\bibitem[\protect\citeauthoryear{Feige, Mirrokni, and
  Vondr{\'{a}}k}{2011}]{DBLP:journals/siamcomp/FeigeMV11}
Feige, U.; Mirrokni, V.~S.; and Vondr{\'{a}}k, J.
\newblock 2011.
\newblock Maximizing non-monotone submodular functions.
\newblock {\em {SIAM} Journal of Computing} 40(4):1133--1153.

\bibitem[\protect\citeauthoryear{Goemans and Williamson}{1995}]{GoemansW95}
Goemans, M.~X., and Williamson, D.~P.
\newblock 1995.
\newblock Improved approximation algorithms for maximum cut and satisfiability
  problems using semidefinite programming.
\newblock {\em Journal of the {ACM}} 42(6):1115--1145.

\bibitem[\protect\citeauthoryear{Gurobi~Optimization}{2018}]{gurobi}
Gurobi~Optimization, L.
\newblock 2018.
\newblock Gurobi optimizer reference manual.

\bibitem[\protect\citeauthoryear{Jegelka and Bilmes}{2011}]{5995589}
Jegelka, S., and Bilmes, J.
\newblock 2011.
\newblock Submodularity beyond submodular energies: Coupling edges in graph
  cuts.
\newblock In {\em Proc. of CVPR},  1897--1904.

\bibitem[\protect\citeauthoryear{Krause, Singh, and
  Guestrin}{2008}]{DBLP:journals/jmlr/KrauseSG08}
Krause, A.; Singh, A.~P.; and Guestrin, C.
\newblock 2008.
\newblock Near-optimal sensor placements in gaussian processes: Theory,
  efficient algorithms and empirical studies.
\newblock {\em Journal of Machine Learning Research} 9:235--284.

\bibitem[\protect\citeauthoryear{Lawrence, Seeger, and
  Herbrich}{2002}]{DBLP:conf/nips/LawrenceSH02}
Lawrence, N.~D.; Seeger, M.~W.; and Herbrich, R.
\newblock 2002.
\newblock Fast sparse gaussian process methods: The informative vector machine.
\newblock In {\em Proc. of {NIPS}},  609--616.

\bibitem[\protect\citeauthoryear{Lee \bgroup et al\mbox.\egroup
  }{2009}]{DBLP:conf/stoc/LeeMNS09}
Lee, J.; Mirrokni, V.~S.; Nagarajan, V.; and Sviridenko, M.
\newblock 2009.
\newblock Non-monotone submodular maximization under matroid and knapsack
  constraints.
\newblock In {\em Proc. of {STOC}},  323--332.

\bibitem[\protect\citeauthoryear{Lin and
  Bilmes}{2010}]{Lin:2010:MSV:1857999.1858133}
Lin, H., and Bilmes, J.
\newblock 2010.
\newblock Multi-document summarization via budgeted maximization of submodular
  functions.
\newblock In {\em Proc. of {HLT}},  912--920.

\bibitem[\protect\citeauthoryear{Maehara \bgroup et al\mbox.\egroup
  }{2017}]{MaeharaKSTK17}
Maehara, T.; Kawase, Y.; Sumita, H.; Tono, K.; and Kawarabayashi, K.
\newblock 2017.
\newblock Optimal pricing for submodular valuations with bounded curvature.
\newblock In {\em Proc. of {AAAI}},  622--628.

\bibitem[\protect\citeauthoryear{Nemhauser and
  Wolsey}{1978}]{DBLP:journals/mor/NemhauserW78}
Nemhauser, G.~L., and Wolsey, L.~A.
\newblock 1978.
\newblock Best algorithms for approximating the maximum of a submodular set
  function.
\newblock {\em Mathematics of Operations Research} 3(3):177--188.

\bibitem[\protect\citeauthoryear{Newman}{2003}]{DBLP:journals/siamrev/Newman03}
Newman, M. E.~J.
\newblock 2003.
\newblock The structure and function of complex networks.
\newblock {\em {SIAM} Review} 45(2):167--256.

\bibitem[\protect\citeauthoryear{Rossi and Ahmed}{2016}]{network-repo}
Rossi, R.~A., and Ahmed, N.~K.
\newblock 2016.
\newblock An interactive data repository with visual analytics.
\newblock {\em SIGKDD Explorations} 17(2):37--41.

\bibitem[\protect\citeauthoryear{Sebastiani and
  Wynn}{2002}]{doi:10.1111/1467-9868.00225}
Sebastiani, P., and Wynn, H.~P.
\newblock 2002.
\newblock Maximum entropy sampling and optimal bayesian experimental design.
\newblock {\em Journal of the Royal Statistical Society: Series B (Statistical
  Methodology)} 62(1):145--157.

\bibitem[\protect\citeauthoryear{Singh \bgroup et al\mbox.\egroup
  }{2009}]{DBLP:journals/jair/SinghKGK09}
Singh, A.; Krause, A.; Guestrin, C.; and Kaiser, W.~J.
\newblock 2009.
\newblock Efficient informative sensing using multiple robots.
\newblock {\em Journal of Artificial Intelligence Research} 34:707--755.

\bibitem[\protect\citeauthoryear{Sviridenko, Vondr{\'{a}}k, and
  Ward}{2017}]{DBLP:journals/mor/SviridenkoVW17}
Sviridenko, M.; Vondr{\'{a}}k, J.; and Ward, J.
\newblock 2017.
\newblock Optimal approximation for submodular and supermodular optimization
  with bounded curvature.
\newblock {\em Mathematics of Operations Research} 42(4).

\bibitem[\protect\citeauthoryear{Tschiatschek, Singla, and
  Krause}{2017}]{TSK17}
Tschiatschek, S.; Singla, A.; and Krause, A.
\newblock 2017.
\newblock Selecting sequences of items via submodular maximization.
\newblock In {\em Proc. of {AAAI}},  2667--2673.

\bibitem[\protect\citeauthoryear{Vondr{\'{a}}k}{2010}]{V10}
Vondr{\'{a}}k, J.
\newblock 2010.
\newblock Submodularity and curvature: the optimal algorithm.
\newblock {\em RIMS Kokyuroku Bessatsu} B23:253--266.

\end{thebibliography}

\onecolumn

\section*{Appendix (Missing Proofs)}
\begin{proof}[Proof of Reduction \ref{reduction}]
Fix a constant $c > 0. $We observe that if the condition of the statement does not hold, then it is sufficient to add a set $D$ of $\sum_i cd_i$ " dummy" elements that do not have any effect on the $f$-values, and remove them from the output of the algorithm, for all $i = 1, \dots , k$. Denote with $D_1, \dots, D_k$ a partition of $D$ with $\absl{D_i} = cd_i$. This only increases the size of the instance by a multiplicative constant factor. We define new subsets $\overline{B}_i = B_i \cup D_i$ for all $i = 1, \dots, k$. Thus, we can maximize the function $f$ on the newly-defined partition constraint without affecting neither the global optimum, nor the value of the algorithm's output.
\end{proof}
\begin{proof}[Proof of Proposition \ref{prop:sum_curvature}]
Fix two subsets $S, \Omega \subseteq V$ of size at most $d$, and a point $\omega \in S \setminus \Omega$. From the definition of curvature we have
\begin{align*}
& f(S \cup \Omega) - f((S \cup \Omega) \setminus \{\omega \}) \geq (1 - \alpha_1) (f(S) - f(S \setminus \{ \omega \})),\\
& g(S \cup \Omega) - g((S \cup \Omega) \setminus \{\omega \}) \geq (1 - \alpha_2) (g(S) - g(S \setminus \{ \omega \})).
\end{align*}
Thus, we have that it holds
\begin{align*}
(f+ g)(S \cup \Omega) & - (f + g)((S \cup \Omega) \setminus \{\omega \})\\
& = f(S \cup \Omega) - f((S \cup \Omega) \setminus \{\omega \}) + g(S \cup \Omega) - g((S \cup \Omega) \setminus \{\omega \})\\
& \geq (1 - \alpha_1) (f(S) - f(S \setminus \{ \omega \})) + (1 - \alpha_2) (g(S) - g(S \setminus \{ \omega \}))\\
& \geq (1 - \sup \alpha_i) \left ( f(S) - f(S \setminus \{ \omega \}) + g(S) - g(S \setminus \{ \omega \}) \right )\\
& = (1 - \sup \alpha_i)\left ((f + g)(S) - (f + g)(S \setminus \{ \omega \}) \right ),
\end{align*}
as claimed.
\end{proof}
\begin{proof}[Proof of Proposition \ref{prop:curvature}]
Fix any two subsets $T, \Omega \subseteq V$, and let $\omega \in S \setminus \Omega$. Then, it holds
\begin{align*}
\alpha & \leq 1 - \min_{\{S\subseteq V,\ \omega \in S\}}\frac{\rho_{\omega}(S)}{\rho_{\omega}(\emptyset)} \leq 1 - \frac{\rho_{\omega}(T\cup \Omega)}{\rho_{\omega}(\emptyset)} \leq 1 - \frac{\rho_{\omega}(T\cup \Omega)}{\rho_{\omega}(T)},
\end{align*}
where the last inequality follows from the definition of submodular function.
\end{proof}
\begin{proof}[Proof of Theorem \ref{thm:det_greedy_algorithm}]
We assume without loss of generality that $f$ is non-constant. Moreover, due to Reduction \ref{reduction}, we may assume that $d > \alpha$. We perform the analysis until a solution of size $\overline{d}$ is found. This is not restrictive, since due to Reduction \ref{reduction}, the value $f(S_t)$ never decreases, for increasing $t$. Let $D$ be a set of dummy elements as in Reduction \ref{reduction}. Let $M$ be a set of size $\absl{M} = d$ of the form $M = \opt \cup \overline{D}$ with $\overline{D} \subseteq D$. We have that it holds
\begin{align}
\rho_t \geq \frac{1}{\absl{M \setminus S_{t - 1}}} \sum_{\omega \in M \setminus S_{t - 1}} \rho_{\omega} (S_{t - 1}) \geq \frac{\rho_{\opt}(S_{t - 1})}{\absl{M \setminus S_{t - 1}}},\label{eq:thm_negative_curv1}
\end{align}
where we have used that $S_{t - 1} \cup \omega$ is always a feasible solution, since $\absl{S_{t - 1} \cup \omega}<\overline{d}$ and $(S_{t - 1} \cup \omega) \subseteq \cup_j B_j$ for all $\omega \in \cup_j B_j$, and the second inequality follows from submodularity, together with the fact that $\absl{\opt} \leq d$. To continue, we consider the following lemma.
\begin{lemma}
\label{lemma:marginal_value}
Following the notation introduced above, define the set $J = \{t \in [\overline{d}] \colon \omega_t \in \opt \}$. Then it holds
\[
f(\omega_t \cup \opt) \geq f(\opt) + (1 - \alpha) \sum_{m \in [t]}\rho_m  - (1 - \alpha)\sum_{m \in [t]\cap J}\rho_m,
\]
for all $1 < t \leq \overline{d}$.
\end{lemma}
\begin{proof}
From the definition of curvature we have that $\rho_{\omega_m}(S_{t-1} \cup \opt) \geq (1 - \alpha) \rho_{m} $ for all $m \in [t] \setminus J$, and $\rho_{\omega_t}(S_{t-1} \cup \opt) \geq (1 - \alpha) \rho_{m} - (1 - \alpha) \rho_{m}$ for all $m \in J$. The claim follows by applying these two inequalities iteratively to the $f(S_{t - 1} \cup \opt)$.
\end{proof}
Define $x_{t} = \rho_t / \opt$ for all $t \in [\overline{d}]$. Note that it holds $d x_{1} \geq 1$. Furthermore, following the notation of Lemma \ref{lemma:marginal_value} we have that it holds $\absl{M \setminus S_t} = d - \absl{[t] \cap J}$, for all $t \in [\overline{d}]$. Combining Lemma \ref{lemma:marginal_value} with \eqref{eq:thm_negative_curv1} we get
\begin{equation}
\label{eq:linear_system1} 
(d - \absl{[t - 1]\cap J}) x_{t} + \alpha \sum_{m \in [t - 1]}x_{m} + (1 - \alpha)\sum_{m \in [t - 1]\cap J }x_{m} \geq 1.
\end{equation}
for all $1 < t\leq [\overline{d}]$. Note that the equations as in \eqref{eq:linear_system1} give an LP of the form\vspace{20pt}\\
$\quad $
\begin{equation*}
\quad \\
\quad \\
\quad \\
\quad \\
\begin{bmatrix}
    d      & & & & & & \\
    \alpha & (d - \absl{[1]\cap J}) & & & & & \\
    \alpha & 1 & (d - \absl{[2]\cap J}) & & & & \\
    \alpha & 1 & 1 & & & & \\
    \vdots & \vdots & \vdots & & & & \\
    \vdots & \vdots & \vdots & & & & \\
    \vdots & \vdots & \vdots & & & & \\
    \vdots & \vdots & \vdots & & & & \\
    \alpha & 1 & 1 & \cdots \quad & & (d - \absl{[\overline{d} - 1]\cap J}) &  \\   
    \alpha & 1 & 1 & \cdots \quad & & \alpha & (d- \absl{[\overline{d}]\cap J})
\end{bmatrix}
\begin{bmatrix}
    x_{1} \\
    x_{2} \\
    x_{3} \\
    \vdots \\
    \vdots \\
    \vdots \\ 
    \vdots \\ 
    x_{\overline{d} - 1} \\
    x_{\overline{d}} 
\end{bmatrix}
\geq
\begin{bmatrix}
    1 \\
    1 \\
    1 \\
    \vdots \\
    \vdots \\
    \vdots \\
    \vdots \\ 
    1 \\
    1 \\
\end{bmatrix}
\end{equation*}
\vspace{20pt}\\
that is completely determined by $\alpha$, $d$, $\overline{d}$ and $J$. We denote with $R(\alpha, d,\overline{d}, J)$ any such LP. In the following, we say that an array $(y_1, \dots, y_{\overline{d}}) = \overline{y} \in \mathbb{R}^{\overline{d}}_{\geq 0}$ is an optimal solution for $R(\alpha, d, \overline{d}, J)$ if $\overline{y}$ is feasible for $R(\alpha, d, \overline{d}, J)$, and if the sum $\sum_m y_m$ is minimal over all feasible solutions of $R(\alpha, d, \overline{d}, J)$. To continue with the proof we consider the following lemma.
\begin{lemma}
\label{lemma:triangular_matrix}
Following the notation introduced above, let $(y_1, \dots, y_{\overline{d}}) = \overline{y} \in \mathbb{R}^{\overline{d}}_{\geq 0}$ be an optimal solution of $R(\alpha, d, \overline{d}, J)$, and suppose that $J \neq \emptyset$. Then it holds $y_q \leq y_{q + 1}$ for all $q \in J$.
\end{lemma}
\begin{proof}
We proceed \emph{ad Absurdum}, by assuming that there exists a point $q \in J$ s.t. $y_{q} > y_{q + 1}$. Define the positive constant
\[
\varepsilon = \frac{d - \absl{[q] \cap J}}{d - \absl{[q - 1] \cap J}}(y_q - y_{q + 1}).
\]
Consider a vector $(z_1, \dots, z_{\overline{d}_i}) \in \mathbb{R}^{\overline{d}}_{\geq 0}$, defined as
\[
\left \{
\begin{array}{ll}
z_m = y_m & \mbox{ if } 1 \leq m < q;\\
z_m = y_m - \varepsilon & \mbox{ if } m = q;\\
z_m = y_m + \frac{\varepsilon}{d - \absl{[q] \cap J}} & \mbox{ if } q < m \leq \overline{d};\\
\end{array}
\right .
\]
We first observe that $(z_1, \dots, z_{\overline{d}})$ is a feasible solution for $R(\alpha, d, \overline{d}, J)$. Note that this is clearly the case for all coefficients $z_m$ with $1 \leq m < q$. Given a vector $(w_1, \dots, w_{\overline{d}})= \overline{w} \in \mathbb{R}^{\overline{d}}_{\geq 0}$, we define
\begin{equation}
\label{eq:system_line}
L_{s}(\overline{w}) = (d - \absl{[s - 1] \cap J}) w_{s} + \alpha \sum_{m \in [s - 1]} w_m + (1 - \alpha)\sum_{m \in [s - 1] \cap J} w_m.
\end{equation}
Essentially, $L_s(\overline{w})$ returns the value obtained by multiplying the $s$-th row of the constraint matrix of the LP $R(\alpha, d, \overline{d}, J)$ with the vector $\overline{w}$. Following this notation, we have that it holds
\begin{align*}
L_{q + 1}(\overline{y}) - L_{q}(\overline{z}) = (d - \absl{[q] \cap J}) y_{q+1} - (d - \absl{[q - 1] \cap J)}) z_q + y_q = 0.
\end{align*}
Since $(y_1, \dots, y_{\overline{d}})$ is a feasible solution, then $L_{q}(\overline{z}) = L_{q + 1}(\overline{y}) \geq 1$. Hence, the coefficients $z_1, \dots, z_q$ are feasible for the system $R(\alpha, d, \overline{d}, J)$. We now prove that the solutions $z_{q + 1}, \dots, z_{\overline{d}}$ are also feasible coefficients. We now proceed by proving that it holds $L_{s}(\overline{z}) - L_s (\overline{y}) \geq 0$, for all $s \in [\overline{d}] \setminus [q]$ with an induction argument on $s$. For the base case with $s = q + 1$, we have that
\[
L_{q + 1}(\overline{z}) - L_{q + 1}(\overline{y}) = \varepsilon +  (d - \absl{ [q] \cap J}) \frac{\varepsilon}{d - \absl{[q] \cap J}} \geq 0.
\]
For the inductive step, we consider two separate cases.\\ \\
(Case 1: $s \in J$) In this case, from \eqref{eq:system_line} it holds 
\begin{align*}
L_{s + 1}(\overline{w}) - L_{s}(\overline{w}) & = (d - \absl{[s] \cap J}) w_{s + 1} + w_{s} - (\hat{d} - \absl{[s - 1] \cap J}) w_{s}\\
& = (d - \absl{[s] \cap J})(w_{s + 1} - w_s),
\end{align*}
where we have used that $s \in J$. Hence, $\left (L_{s + 1}(\overline{z}) - L_{s}(\overline{z}) \right ) -  \left (L_{s + 1}(\overline{y}) - L_{s}(\overline{y}) \right ) = 0$.\\ \\
(Case 2: $s \notin J$) We use \eqref{eq:system_line} again, to show that it holds 
\begin{align*}
L_{s + 1}(\overline{w}) - L_{s}(\overline{w}) & = (d - \absl{[s] \cap J}) w_{s + 1} + \alpha w_{s} - (d - \absl{[s - 1] \cap J}) w_{s}\\
& = (d - \absl{[s - 1] \cap J})(w_{s + 1} - w_s) + \alpha w_s\\
& \geq (d - \absl{[s - 1] \cap J})(w_{s + 1} - w_s),
\end{align*}
where we have used that $\alpha, w_s \geq 0$. We conclude that $\left (L_{s + 1}(\overline{z}) - L_{s}(\overline{z}) \right ) -  \left (L_{s + 1}(\overline{y}) - L_{s}(\overline{y}) \right ) \geq 0$.\\ \\
Combining the two cases discussed above, we have that
\[
L_{s + 1}(\overline{z}) - L_{s + 1}(\overline{y}) \geq L_s(\overline{z}) - L_s(\overline{y})
\]
for all $s \in [\overline{d}] \setminus [q]$. Therefore, we use thy inductive hypothesis on the $L_s(\overline{z}) - L_s(\overline{y})$ and conclude that the sequence $(z_1, \dots , z_{\overline{d}})$ is a feasible solution for the system $J(\alpha, d, \overline{d}, J)$.\\ \\
We conclude the proof, by showing that $\overline{z}$ contradicts the assumption that $\overline{y}$ is minimal. To this end, we observe that it holds
\[
\sum_{m \in [\overline{d}]} (z_m - y_m)  = \left ( -1 + \frac{\overline{d} - q - 1}{\overline{d} - q}  \right ) \varepsilon \leq 0.
\]
Since the coefficients of the matrix in the linear system $R(\alpha, d, \overline{d}, J)$ are non-negative, this proves the claim. 
\end{proof}
The lemma above is useful, because it allows us to significantly simplify our setting. In fact, using Lemma \ref{lemma:triangular_matrix} we can prove the following result.
\begin{lemma}
\label{lemma:triangluar_matrix_simplified}
Following the notation introduced above, let $(y_1, \dots, y_{\overline{d}}) = \overline{y} \in \mathbb{R}^{\overline{d}}_{\geq 0}$ be an optimal solution of the system $R(\alpha, d, \overline{d}, J)$, and let $(z_1, \dots, z_{\overline{d}}) = \overline{z} \in \mathbb{R}^{\overline{d}}_{\geq 0}$ be an optimal solution of the system $R(\alpha, d, \overline{d}, \emptyset)$. Then it holds
\[
\sum_{m = 1}^{\overline{d}} y_m \geq \sum_{m = 1}^{\overline{d}} z_m.
\]
\end{lemma}
\begin{proof}
Denote with $t_1, \dots t_{\absl{J}}$ the points of $J$ sorted in increasing order, and define the set
\[
J' = \left \{ t_{\absl{J}} - \absl{J} + 1, \dots, t_{\absl{J}} - 1, t_{\absl{J}} \right \}.
\]
Let $(y_1', \dots , y_m') = \overline{y}' \in \mathbb{R}^{\overline{d}}_{\geq 0}$ be an optimal solution of the system $R(\alpha, d, \overline{d}, J')$. We first observe that it holds
\begin{equation}
\label{eq:first_step_of_second_lemma}
\sum_{m = 1}^{\overline{d}} y_m \geq \sum_{m = 1}^{\overline{d}} y'_m.
\end{equation}
To this end, suppose that $J \neq J'$, and let $t_q \in J$ be a point s.t. $t_q + 1 < t_{q + 1}$. Define the set $I = \left \{ t_1, \dots, t_{q-1}, t_{q + 1}, t_{q + 1}, \dots, t_{\absl{J}} - 1, t_{\absl{J}} \right \}$. Using Lemma \ref{lemma:triangular_matrix}, we observe that the solution $\overline{y}$ is also feasible for the system $R(\alpha, d, \overline{d}, I')$, and \eqref{eq:first_step_of_second_lemma} easily follows from this observation.\\ \\
Define $I' = \left \{ t_{\absl{J}} - \absl{J} + 2, \dots, t_{\absl{J}} - 1, t_{\absl{J}} \right \}$. Note that it holds $I'  = J' \setminus \left \{ t_{\absl{J}} - \absl{J} + 1 \right \}$. Using Lemma \ref{lemma:triangular_matrix}, we have that it holds $y'_{t_{\absl{J}} - \absl{J} + 1} \leq \dots \leq y'_{t_{\absl{J}}}$, and the solution $\overline{y}'$ is feasible for the system $R(\alpha, d, \overline{d}, I')$. We can iterate this process, to conclude that $\overline{y}'$ is a feasible solution for $R(\alpha, d, \overline{d}, \emptyset)$. The claim follows combining this observation with \eqref{eq:first_step_of_second_lemma}.
\end{proof}
Lemma \eqref{lemma:triangluar_matrix_simplified} is very useful in that it allows us to significantly simplify our setting. In fact, we can obtain the desired approximation guarantee by studying the following LP
\begin{equation}
\label{eq:lp}
\begin{bmatrix}
    d & 0 & \hdotsfor{4}  & 0 \\
    \alpha & d & 0 & \hdotsfor{3}  & 0 \\
    \alpha & \alpha & d & 0 & \hdotsfor{2}  & 0 \\
    \vdots & \vdots & \vdots & \ddots & \ddots &  & \vdots \\
    \alpha & \hdotsfor{2} & \alpha & d & 0 & 0 \\
    \alpha & \hdotsfor{3} & \alpha & d & 0 \\
    \alpha & \hdotsfor{4} & \alpha & d \\
\end{bmatrix}
\begin{bmatrix}
    x_{i, 1} \\
    x_{i, 2} \\
    x_{i, 3} \\
    \vdots \\
    x_{i, \overline{d} - 2} \\
    x_{i, \overline{d} - 1} \\
    x_{\overline{d}} \\
\end{bmatrix}
\geq
\begin{bmatrix}
    1 \\
    1 \\
    1 \\
    \vdots \\
    1 \\
    1 \\
    1 \\
\end{bmatrix},
\end{equation}
which corresponds to the case $R(\alpha, d, \overline{d}, \emptyset)$. To continue with the proof, we consider the following lemma.
\begin{lemma}
\label{lemma:general_curvature}
Let $(y_1, \dots , y_{\overline{d}}) \in \mathbb{R}^{\overline{d}}_{\geq 0}$ be a solution to the LP given in \eqref{eq:lp}. Then it holds
\[
y_t \geq \frac{1}{d} \left ( 1 - \frac{\alpha}{d} \right )^{t - 1}
\]
for all $t = 1, \dots, \overline{d}$.
\end{lemma}
\begin{proof}
We first show by induction that any solution $(z_1, \dots, z_{\overline{d}})$ that fulfills the constrains
\begin{equation}
\label{eq:lp2}
\begin{bmatrix}
    d & 0 & \hdotsfor{4} & 0 \\
    \alpha & d & 0 & \hdotsfor{3}  & 0 \\
    \alpha & \alpha & d & 0 & \hdotsfor{2}  & 0 \\
    \vdots & \vdots & \vdots & \ddots & \ddots &  & \vdots \\
    \alpha & \hdotsfor{2} & \alpha & d & 0 & 0 \\
    \alpha & \hdotsfor{3} & \alpha & d & 0 \\
    \alpha & \hdotsfor{4} & \alpha & d \\
\end{bmatrix}
\begin{bmatrix}
    z_1 \\
    z_2 \\
    z_3 \\
    \vdots \\
    z_{\overline{d} - 2} \\
    z_{\overline{d} - 1} \\
    z_{\overline{d}} \\
\end{bmatrix}
=
\begin{bmatrix}
    1 \\
    1 \\
    1 \\
    \vdots \\
    1 \\
    1 \\
    1 \\
\end{bmatrix}
\end{equation}
yields
\begin{equation}
\label{eq:inductive_claim}
z_t = \frac{1}{d} \left ( 1 - \frac{\alpha}{d} \right )^{t - 1}.
\end{equation}
The base case with $t = 1$ is trivially true. Suppose now that the claim holds for all $z_1, \dots, z_{t - 1}$. Then it holds
\begin{equation*}
z_t = \frac{1}{d} \left (1 - \alpha \sum_{j = 1}^{t - 1} z_j \right )  = \frac{1}{d} \left ( 1 - \alpha \sum_{j = 1}^{t - 1} \frac{1}{d} \left ( 1 - \frac{\alpha}{d} \right )^{j - 1} \right )  = \frac{1}{d} \left ( 1 - \frac{\alpha}{d} \right )^{t - 1},
\end{equation*}
and \eqref{eq:inductive_claim} holds. In particular, since any other solution $(y_1, \dots , y_k)$ to \eqref{eq:lp} is s.t. $y_t \geq z_t$ for all $t = 1, \dots, \overline{d}$, and since $t_t \geq z_t$ for all $t = 1, \dots, \overline{d}$, then the claim follows.
\end{proof}
Thus, combining Lemma \ref{lemma:general_curvature} with Lemma \ref{lemma:triangluar_matrix_simplified} it holds
\begin{equation}
\label{eq:positive_curvature3}
\rho_t \geq \frac{1}{d} \left ( 1 - \frac{\alpha}{d} \right )^{t - 1}f(\opt)
\end{equation}
for all $t = 1, \dots, \overline{d}$. Therefore, we have that
\[
f\left (S_{\overline{d}} \right ) = \sum_{t = 1}^{\overline{d}} \rho_t \geq \sum_{t = 1}^{\overline{d}} \frac{1}{d} \left ( 1 - \frac{\alpha}{d} \right )^{t - 1} f(\opt) = \frac{1}{\alpha} \left (1 - \left ( 1 - \frac{\alpha}{d} \right )^{\overline{d}} \right )f(\opt) \geq \frac{1}{\alpha} \left (1 - e^{-\alpha \overline{d}/d} \right )f(\opt),
\]
where we have used \eqref{eq:positive_curvature3}.
\end{proof}
\begin{proof}[Proof of Theorem \ref{thm:submodular_function_positive2}]
This proof is similar to that of Theorem \ref{thm:det_greedy_algorithm}. Again, we assume without loss of generality that $f$ is non-constant, and we preform the analysis until a solution of size $\overline{d}$ is found. We first observe that the following holds. Let $S\subseteq V$ be any subset of size at most $d$ such that $\rho_{\opt}(S) \ne 0$. Then it holds
\begin{equation}
\label{eq:thm_negative_curv0_subadditive}
\sum_{\omega \in \opt \setminus S}\frac{\rho_{\omega}(S)}{\rho_{\opt }(S)} \geq \sum_{\omega \in \opt \setminus S}\frac{\rho_{\omega}(S)}{f(\opt \setminus S)} \geq (1 - \alpha ) \sum_{\omega \in \opt \setminus S}\frac{f(\omega )}{f(\opt \setminus S)} \geq (1 - \alpha ),
\end{equation}
where we have used the definition of subadditivity. Let $M$ be a set of size $\absl{M} = d$ of the form $M = \opt \cup \overline{D}$ with $\overline{D} \subseteq D$. We have that it holds
\begin{align}
\rho_t \geq \frac{1}{\absl{M \setminus S_{t - 1}}} \sum_{\omega \in M \setminus S_{t - 1}} \rho_{\omega} (S_{t - 1}) \geq (1 - \alpha)\frac{\rho_{\opt}(S_{t - 1})}{\absl{M \setminus S_{t - 1}}},\label{eq:thm_negative_curv1_subadditive}
\end{align}
where we have used that $S_{t - 1} \cup \omega$ is always a feasible solution, since $\absl{S_{t - 1} \cup \omega}<\overline{d}$ and $(S_{t - 1} \cup \omega) \subseteq \cup_j B_j$ for all $\omega \in \cup_j B_j$, and the second inequality follows from \eqref{eq:thm_negative_curv0_subadditive}, together with the fact that $\absl{\opt} \leq d$.
Define $x_{t} = \rho_t / \opt$ for all $t \in [\overline{d}]$. Note that it holds $d x_{1} \geq 1$. Furthermore, defining $J = \{t \in [\overline{d}] \colon \omega_t \in \opt \}$, we have that it holds $\absl{M \setminus S_t} = d - \absl{[t] \cap J}$, for all $t \in [\overline{d}]$. Combining Lemma \ref{lemma:marginal_value} with \eqref{eq:thm_negative_curv1} we get
\begin{equation}
\label{eq:linear_system1} 
\frac{d - \absl{[t - 1]\cap J}}{1 - \alpha} x_{t} + \alpha \sum_{m \in [t - 1]}x_{m} + (1 - \alpha)\sum_{m \in [t - 1]\cap J }x_{m} \geq 1.
\end{equation}
for all $1 < t\leq [\overline{d}]$. Note that the equations as in \eqref{eq:linear_system1} give an LP of the form\vspace{20pt}\\
$\quad $
\begin{equation*}
\quad \\
\quad \\
\quad \\
\quad \\
\begin{bmatrix}
    d      & & & & & & \\
    \alpha & \frac{d - \absl{[1]\cap J}}{1 - \alpha} & & & & & \\
    \alpha & 1 & \frac{d - \absl{[2]\cap J}}{1 - \alpha} & & & & \\
    \alpha & 1 & 1 & & & & \\
    \vdots & \vdots & \vdots & & & & \\
    \vdots & \vdots & \vdots & & & & \\
    \vdots & \vdots & \vdots & & & & \\
    \vdots & \vdots & \vdots & & & & \\
    \alpha & 1 & 1 & \cdots \quad & & \frac{d - \absl{[\overline{d} - 1]\cap J}}{1 - \alpha} &  \\   
    \alpha & 1 & 1 & \cdots \quad & & \alpha & \frac{d- \absl{[\overline{d}]\cap J}}{1 - \alpha}
\end{bmatrix}
\begin{bmatrix}
    x_{1} \\
    x_{2} \\
    x_{3} \\
    \vdots \\
    \vdots \\
    \vdots \\ 
    \vdots \\ 
    x_{\overline{d} - 1} \\
    x_{\overline{d}} 
\end{bmatrix}
\geq
\begin{bmatrix}
    1 \\
    1 \\
    1 \\
    \vdots \\
    \vdots \\
    \vdots \\
    \vdots \\ 
    1 \\
    1 \\
\end{bmatrix}
\end{equation*}
\vspace{20pt}\\
that is completely determined by $\alpha$, $d$, $\overline{d}$ and $J$. Again, we denote with $R(\alpha, d,\overline{d}, J)$ any such LP. As in the proof of Theorem \ref{thm:det_greedy_algorithm}, we say that an array $(y_1, \dots, y_{\overline{d}}) = \overline{y} \in \mathbb{R}^{\overline{d}}_{\geq 0}$ is an optimal solution for $R(\alpha, d, \overline{d}, J)$ if $\overline{y}$ is feasible for $R(\alpha, d, \overline{d}, J)$, and if the sum $\sum_m y_m$ is minimal over all feasible solutions of $R(\alpha, d, \overline{d}, J)$. The remaining part of the proof follows along the lines of that of Theorem \ref{thm:det_greedy_algorithm}. To continue, we use the following lemma.
\begin{lemma}
\label{lemma:triangular_matrix_subadditive}
Following the notation introduced above, let $(y_1, \dots, y_{\overline{d}}) = \overline{y} \in \mathbb{R}^{\overline{d}}_{\geq 0}$ be an optimal solution of $R(\alpha, d, \overline{d}, J)$, and suppose that $J \neq \emptyset$. Then it holds $y_q \leq y_{q + 1}$ for all $q \in J$.\footnote{Note that this statement is identical to that of Lemma \ref{lemma:triangular_matrix}. However, due to the fact that $R(\alpha, d, \overline{d}, J)$ is defined differently, it requires a more involved proof.}
\end{lemma}
\begin{proof}
We proceed \emph{ad Absurdum}, by assuming that there exists a point $q \in J$ s.t. $y_{q} > y_{q + 1}$. Consider the positive constant
\[
\varepsilon = \frac{(d - \absl{[q] \cap J} + \alpha)y_q - (d - \absl{[q] \cap J})y_{q + 1}}{d - \absl{[q - 1] \cap J}},
\]
and define
\[
\left \{ 
\begin{array}{lll}
\varepsilon_1 = \varepsilon \left ( \frac{1 - \alpha}{d - \absl{[q] \cap J}} \right ) & \mbox{if} & m = 1; \\
\varepsilon_m = \varepsilon_{m - 1} \left ( 1 + \frac{\alpha}{d - \absl{[q ] \cap J} - m + q} \right ) & \mbox{if} & 1 < m < \overline{d} - q; \\
\end{array}
\right .
\]
Consider a vector $(z_1, \dots, z_{\overline{d}_i}) \in \mathbb{R}^{\overline{d}}_{\geq 0}$, defined as
\[
\left \{
\begin{array}{lll}
z_m = y_m & \mbox{if} & 1 \leq m < q;\\
z_m = y_m - \varepsilon & \mbox{if} & m = q;\\
z_m = y_m + \varepsilon_{m - q} & \mbox{if} & q + 1 \leq m \leq \overline{d};\\
\end{array}
\right .
\]
Again, we proceed by showing that $(z_1, \dots, z_{\overline{d}})$ is a feasible solution for $R(\alpha, d, \overline{d}, J)$, and that this leads to a contradiction. 
Note that this is clearly the case for all coefficients $z_m$ with $1 \leq m < q$. Given a vector $(w_1, \dots, w_{\overline{d}})= \overline{w} \in \mathbb{R}^{\overline{d}}_{\geq 0}$, we define
\begin{equation}
\label{eq:system_line}
L_{s}(\overline{w}) = \frac{d - \absl{[s - 1] \cap J}}{1 - \alpha} w_{s} + \alpha \sum_{m \in [s - 1]} w_m + (1 - \alpha)\sum_{m \in [s - 1] \cap J} w_m.
\end{equation}
Again, $L_s(\overline{w})$ returns the value obtained by multiplying the $s$-th row of the constraint matrix of the LP $R(\alpha, d, \overline{d}, J)$ with the vector $\overline{w}$.\\
Following this notation, one can easily verify as in Lemma \ref{lemma:triangular_matrix} that it holds $L_{q - 1}(\overline{z}) \geq 1$ and $L_{q}(\overline{z}) \geq 1$.\\
We now prove that $L_{s}(\overline{z}) \geq 1$, for all $s > q$.\\
The base case with $s = q + 1$ can be verified directly. For the inductive case, suppose that it holds $L_{q + u}(\overline{z}) \geq 1$ for some $u > 1$. We distinguish two cases. \\ \\
(Case 1) $q + u \in J$. In this case, we have that it holds 
\begin{align*}
L_{q + u + 1}(\overline{y}) & - L_{q + u + 1}(\overline{z}) - L_{q + u}(\overline{y}) + L_{q + u}(\overline{z}) \\
& = \frac{d - \absl{[q + u] \cap J}}{1 - \alpha} \varepsilon_{u + 2} - \frac{d - \absl{[q + u] \cap J} + \alpha}{1 - \alpha} \varepsilon_{u + 1} \\
& = \left ((d - \absl{[q + u] \cap J})\frac{d - \absl{[q]\cap J} - u - 2  + \alpha}{d - \absl{[q]\cap J} - u - 2} - (d - \absl{[q + u] \cap J} + \alpha ) \right ) \frac{\varepsilon_{u + 1}}{1 - \alpha} \\
& \geq \left ((d - \absl{[q + u] \cap J})\frac{d - \absl{[q]\cap J} - u - 2  + \alpha}{d - \absl{[q]\cap J} - u - 2} - (d - \absl{[q + u] \cap J} + \alpha) \right ) \frac{\varepsilon_{u + 1}}{1 - \alpha}\\
& = 0.
\end{align*}
Hence, we conclude that $L_{q + u + 1}(\overline{z}) \geq 0$. \\ \\
(Case 2) $q + u \notin J$. We have that it holds
\begin{align*}
L_{q + u + 1}(\overline{y}) & - L_{q + u + 1}(\overline{z}) - L_{q + u}(\overline{y}) + L_{q + u}(\overline{z}) \\
& = \frac{d - \absl{[q + u] \cap J}}{1 - \alpha} \varepsilon_{u + 2} - \left ( \frac{d - \absl{[q + u] \cap J}}{1 - \alpha} - \alpha \right ) \varepsilon_{u + 1} \\
& \geq \frac{d - \absl{[q + u] \cap J}}{1 - \alpha}( \varepsilon_{u + 2} - \varepsilon_{u + 1}) \\
& \geq 0.
\end{align*}
Hence, it follows that $L_{q + u + 1}(\overline{z})$ is a feasible solution in this case as well.\\ \\

We conclude that the solution $\overline{z}$ is a feasible solution. We now prove that $\overline{z}$ contradicts the minimality of $\overline{y}$. To this end, we prove that it holds $\sum_m y_m - \sum_m z_m \leq 0$. We have that it holds
\begin{align*}
\sum_m y_m - \sum_m z_m & = \varepsilon \left ( - 1 + \left ( \frac{1 - \alpha}{d - \absl{[q] \cap J}} \right ) + \left ( \frac{1 - \alpha}{d - \absl{[q] \cap J}} \right ) \sum_{m = 2}^{\overline{d} - q}\prod_{\ell = 2}^{m}\left ( 1 + \frac{\alpha}{d - \absl{[q ] \cap J} - \ell + q} \right ) \right ) \\
& \leq \varepsilon \left ( - 1 + \left ( \frac{1 - \alpha}{d - \absl{[q]}} \right ) + \left ( \frac{1 - \alpha}{d - \absl{[q] }} \right ) \sum_{m = 2}^{\overline{d} - q}\prod_{\ell = 2}^{m}\left ( 1 + \frac{\alpha}{d - \absl{[q ]} - \ell + q} \right ) \right ) \\
& = \varepsilon \left ( - 1 + \left ( \frac{1 - \alpha}{d - q} \right ) + \left ( \frac{1 - \alpha}{d - q} \right ) \sum_{m = 2}^{\overline{d} - q}\prod_{\ell = 2}^{m}\left ( 1 + \frac{\alpha}{d - \ell} \right ) \right ) \\
& = \varepsilon \left ( - 1 + \left ( \frac{1 - \alpha}{d - q} \right ) + \left ( \frac{1 - \alpha}{d - q} \right ) \sum_{m = 2}^{\overline{d} - q}\prod_{\ell = 2}^{m}\left ( 1 + \frac{\alpha}{d - \ell} \right ) \right ).
\end{align*}
By computing the sum in the last inequality, and considering the worst-case with $\alpha = 0$, we obtain the claim.
\end{proof}
Again, we combine Lemma \ref{lemma:triangular_matrix_subadditive} with Lemma \ref{lemma:triangluar_matrix_simplified} to simplify our setting. In fact, we can obtain the desired approximation guarantee by studying the following LP
\begin{equation}
\label{eq:lp+}
\begin{bmatrix}
    \frac{d}{1 - \alpha} & 0 & \hdotsfor{4}  & 0 \\
    \alpha & \frac{d}{1 - \alpha} & 0 & \hdotsfor{3}  & 0 \\
    \alpha & \alpha & \frac{d}{1 - \alpha} & 0 & \hdotsfor{2}  & 0 \\
    \vdots & \vdots & \vdots & \ddots & \ddots &  & \vdots \\
    \alpha & \hdotsfor{2} & \alpha & \frac{d}{1 - \alpha} & 0 & 0 \\
    \alpha & \hdotsfor{3} & \alpha & \frac{d}{1 - \alpha} & 0 \\
    \alpha & \hdotsfor{4} & \alpha & \frac{d}{1 - \alpha} \\
\end{bmatrix}
\begin{bmatrix}
    x_1 \\
    x_2 \\
    x_3 \\
    \vdots \\
    x_{d - 2} \\
    x_{d - 1} \\
    x_{\overline{d}} \\
\end{bmatrix}
\geq
\begin{bmatrix}
    1 \\
    1 \\
    1 \\
    \vdots \\
    1 \\
    1 \\
    1 \\
\end{bmatrix}
\end{equation}
which corresponds to the case $R(\alpha, d, \overline{d}, \emptyset)$. As in the proof of Theorem \ref{thm:det_greedy_algorithm}, since $x_t = \rho_t/f(\opt)$, by solving the system \ref{eq:lp+} we get
\begin{equation}
\label{eq:positive_curvature3+}
\rho_t \geq \frac{(1 - \alpha)}{d} \left ( 1 - \frac{\alpha (1 - \alpha)}{d} \right )^{t - 1}f(\opt)
\end{equation}
for all $t = 1, \dots, \overline{d}$. Therefore, we have that
\begin{align*}
f\left (S_{\overline{d}} \right ) = \sum_{t = 1}^{\overline{d}} \rho_t & \geq \sum_{t = 1}^{\overline{d}} \frac{1 - \alpha}{d} \left ( 1 - \frac{\alpha (1 - \alpha)}{d} \right )^{t - 1} f(\opt) \\
& = \frac{1 - \alpha}{\alpha} \left (1 - \left ( 1 - \frac{\alpha(1 - \alpha)}{d} \right )^{\overline{d}} \right )f(\opt) \geq \frac{1}{\alpha} \left (1 - e^{-\alpha (1 - \alpha) \overline{d}/d} \right )f(\opt),
\end{align*}
where we have used \eqref{eq:positive_curvature3+}.
\end{proof}
\end{document}